\documentclass[12pt,onecolumn, twoside]{IEEEtran}
\usepackage{graphicx}
\usepackage{amssymb}
\usepackage{amsmath}
\usepackage{dsfont}
\usepackage{url}
\newtheorem{theorem}{Theorem}

\newtheorem{cor}{Corollary}
\newtheorem{lem}{Lemma}
\newtheorem{prop}[theorem]{Proposition}
\newtheorem{defn}{Definition}

\newcommand{\comb}[2]{\left( \begin{array}{c}  #1 \\ #2 \\  \end{array}  \right)}

\newcommand{\Reals}{\mathds R}
\newcommand{\Integers}{\mathds Z}

\newcommand{\bx}{\mathbf{x}}
\newcommand{\ba}{\mathbf{a}}

\newcommand{\bu}{\mathbf{u}}

\newcommand{\by}{\mathbf{y}}
\newcommand{\bz}{\mathbf{z}}

\newcommand{\bX}{\mathbf{X}}
\newcommand{\bY}{\mathbf{Y}}
\newcommand{\bU}{\mathbf{U}}
\newcommand{\bZ}{\mathbf{Z}}

\newcommand{\eps}{\varepsilon}

\newcommand{\cA}{\mathcal{A}}
\newcommand{\cF}{\mathcal{F}}
\newcommand{\cY}{\mathcal{Y}}
\newcommand{\cX}{\mathcal{X}}
\newcommand{\cZ}{\mathcal{Z}}
\newcommand{\cC}{\mathcal{C}}
\newcommand{\cP}{\mathcal{P}}
\newcommand{\cB}{\mathcal{B}}
\newcommand{\cU}{\mathcal{U}}

\newcommand{\EE}{\mathbb{E}}

\newcommand{\ra}{\rightarrow}

\newcommand{\type}{\mathsf{P}} 
\newcommand{\typecl}{\mathsf{T}} 

\usepackage{color}

\newcommand{\ID}{\mathrm{ID}}
\newcommand{\LCT}{{\ensuremath{\textrm{LC-}\triangle}}}
\newcommand{\TCT}{{\ensuremath{\textrm{TC-}\triangle}}}

\newcommand{\MAYBE}{\mathtt{maybe}}

\newcommand{\NO}{\mathtt{no}}

\newcommand{\Wrev}{P_{X|U}}

\newcommand{\Ber}{\mathrm{Ber}}

\begin{document}

\title{The Minimal Compression Rate \\for Similarity Identification}
\author{Amir Ingber and Tsachy Weissman%
\thanks{The authors are with the Dept. of Electrical Engineering, Stanford University, Stanford, CA 94305.
Email: \{ingber, tsachy\}@stanford.edu.}
\thanks{This work is supported in part by the NSF Center for Science of Information
under grant agreement CCF-0939370, and by a Google research award}
}
\date{\today}
\maketitle
\markboth{Submitted to IEEE Transactions on Information Theory}{Ingber and Weissman: The Minimal Compression Rate for Similarity Identification}
\begin{abstract}
    Traditionally, data compression deals with the problem of concisely representing a data source, e.g. a sequence of letters, for the purpose of eventual reproduction (either exact or approximate). In this work we are interested in the case where the goal is to answer \emph{similarity queries} about the compressed sequence, i.e. to identify whether or not the original sequence is similar to a given query sequence.

We study the fundamental tradeoff between the compression rate and the reliability of the queries performed on compressed data. For i.i.d. sequences, we characterize the minimal compression rate that allows query answers, that are reliable in the sense of having a vanishing false-positive probability, when false negatives are not allowed. The result is partially based on a previous work by Ahlswede et al. \cite{Ahlswede97}, and the inherently typical subset lemma plays a key role in the converse proof.

We then characterize the compression rate achievable by schemes that use lossy source codes as a building block, and show that such schemes are, in general, suboptimal. Finally, we tackle the problem of evaluating the minimal compression rate, by converting the problem to a sequence of convex programs that can be solved efficiently.
\end{abstract}

\section{Introduction}

Traditionally, data compression deals with concisely representing data, e.g., a sequence of letters, for the purpose of eventual reproduction (either exact or approximate). More generally, one wishes to know \emph{something} about the source from its compressed representation. In this work, we are interested in compression when the goal is to identify whether the original source sequence is \emph{similar} to a given query sequence.

A typical scenario where this problem arises is in database queries. Here, a large database containing many sequences $\{\bx_1,...,\bx_M\}$ is required to answer queries of the sort ``what are the sequences in the database that are close to the sequence $\by$?''. Such a scenario (see Fig.~\ref{fig:DB_original}) appears, for example, in computation biology (where the sequences can be, e.g., DNA sequences), forensics (where the sequences represent fingerprints) and internet search.

\begin{figure}[t]
  \centering
  \includegraphics[width=4in]{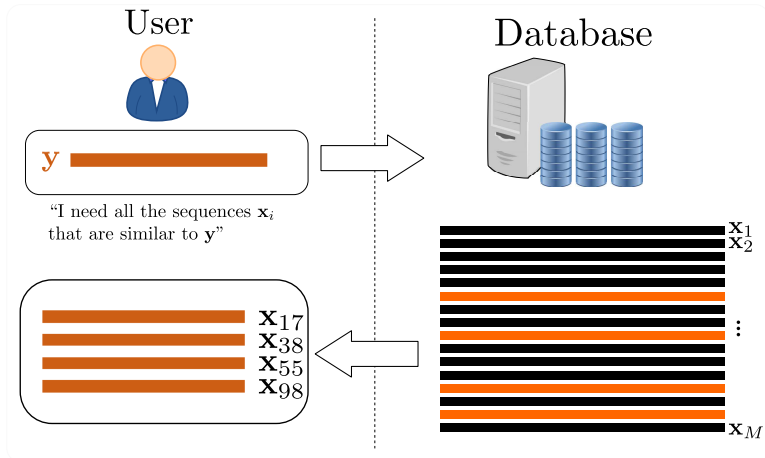}\\
  \caption{Similarity queries in a database.}\label{fig:DB_original}
\end{figure}

Specifically, our interest is in the case where for each sequence $\bx$ in the database we keep a short \emph{signature} $T(\bx)$, and the similarity queries are performed based on $T(\bx)$ and $\by$. Our setting differs from classical compression in that we do not require that the original data be reproducible from the signatures, so the signatures are not meant to replace the original database.
There are many instances where such compression is desirable. For example, the set of signatures can be thought of as a cached version of the original database which, due to its smaller size, can be stored on a faster media (e.g. RAM), or even hosted on many location in order to reduce the burden on the main database. Typically, the user will eventually request the relevant sequences (and only them) from the original database -- see Fig.~\ref{fig:DB_cache}.
\begin{figure}[t]
  \centering
  \includegraphics[width=4in]{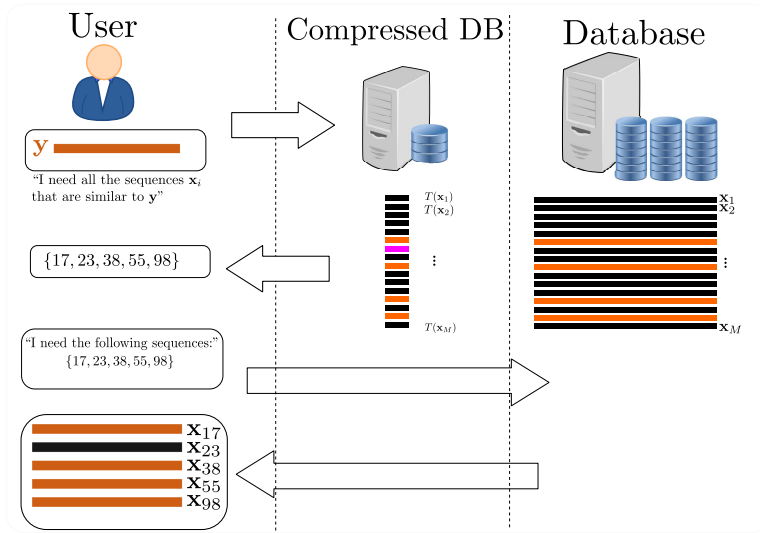}\\
  \caption{Answering similarity queries using a compressed database: first the user receives a set of potential matches, and then asks the original database for the actual sequences. In the example, $\bx_{23}$ is a false positive (FP).
  }\label{fig:DB_cache}
\end{figure}

Naturally, when the queries are answered from the compressed data, one cannot expect to get accurate answers all the time. There are two error events that may occur: the first is a \emph{false positive} (FP), where the query returns a positive answer (``the sequences \emph{are} similar'') but the answer is wrong (the sequences are in fact dissimilar). The second is a \emph{false negative} (FN), when the query returns a negative answer (``the sequences are not similar'') but the sequences are actually similar\footnote{In the statistics literature, the FP and FN events are also known as \emph{type 1} and \emph{type 2} errors, respectively, while in the engineering literature they are
also known as \emph{false alarm} and \emph{misdetection} events, respectively.}. Therefore the interesting tradeoff to consider is between the compression rate (the amount of space required to represent $T(\bx)$) and the reliability of the query answers (measured by the FP and FN probabilities, under an appropriately defined probabilistic model).

The problem was first studied from an information-theoretic viewpoint in the seminal work by Ahlswede et al. \cite{Ahlswede97}. In this work, the source and query sequences are assumed to be drawn i.i.d. from a known distribution, and both false negatives and false positives are allowed. In \cite{Ahlswede97}, the authors first considered the case where the probability of false positives and the probability of false negatives are only required to vanish with the dimension $n$ (in the same spirit of the definition of an achievable rate for communication as a rate for which the error probability can be made to vanish). However, it was shown in \cite{Ahlswede97} that this definition leads to degeneracy: there exist schemes for compression at rates that are arbitrarily close to zero while the two error probabilities vanish. Then, the authors in \cite{Ahlswede97} moved on to consider the case where the FP and FN probabilities are required to vanish \emph{exponentially}, with prescribed exponents $\alpha$ and $\beta$, respectively, and were able to find the optimal compression rate in that setting\footnote{This optimal rate, however, is uncomputable, since the expression depends on an auxiliary random variable with unbounded cardinality.}. We note that this case is atypical in information theory: in the channel coding setting, the highest achievable rate (the channel capacity) is the same, regardless of whether an achievable rate is defined by an error probability vanishing \emph{exponentially} or just vanishing. The same holds for the lowest compression rate for lossy reproduction (the rate-distortion function).

In this paper, we consider the case where no false negatives are allowed. The main motivation is that false negatives cause an \emph{undetected error} in the system where, in contrast, false positives can be easily detected (after retrieving the sequences `flagged' as potential matches, it is easy to filter out the false positives). This is important in several applications, where one cannot compromise on the accuracy of the results (e.g. in a forensic database), but still would like to enjoy the benefits of compression. While it is natural to ask what can be gained when the FN probability is nonzero but `very very small' -- it is important to recall that this probability is based on a probabilistic model of the data, which may not be accurate (in fact, it is rarely the case, especially in source coding settings, where the probabilistic model matches the actual data very closely).

The contributions of the current paper are as follows.
\begin{enumerate}
  \item We find the minimal compression rate for reliable similarity identification with no false negatives. This rate, called the \emph{identification rate} and denoted $R_\ID(D)$, turns out to be the infimal rate at which the ``false-positive'' exponent of \cite{Ahlswede97} is positive. Our result holds for both fixed and variable length compression schemes.

  \item In the case where $\bx$ and $\by$ have the same alphabet, and the similarity measure satisfies the \emph{triangle inequality}, we analyze two schemes for compression that are based on the notion of lossy compression. In those schemes the signature is used for producing a reconstruction $\hat\bx$ of the source, and the decision whether $\bx$ and $\by$ is done according to the distance between $\hat \bx$ and $\by$. We show that those schemes, although simpler for analysis and implementation, attain rates that are generally suboptimal, i.e. strictly greater than $R_\ID(D)$.

  \item The identification rate $R_\ID(D)$ is stated as a non-convex optimization program with an auxiliary random variable. We provide two results that facilitate the computation of $R_\ID(D)$. First, we improve a bound on the cardinality of the auxiliary RV. Then, we propose a method of transforming the said non-convex optimization program into a sequence of \emph{convex} optimization programs, and by that allowing efficient computation of $R_\ID(D)$ for small alphabet sizes. We demonstrate the effectiveness of this approach by calculating $R_\ID(D)$ for several sources.
\end{enumerate}

\medskip

The paper is organized as follows.
In the next section we provide an extended literature survey, which compares the setting discussed in the paper with other ideas, including different hashing schemes. In Sec.~\ref{sec:setting} we formulate the problem, and in Sec.~\ref{sec:MainResults} we state and discuss our main results.
In Sec.~\ref{sec:ProofConverse} we prove the result for the identification rate,
Sec.~\ref{sec:Triangle} contains the analysis of the schemes based on lossy compression, and in Sec.~\ref{sec:computation} we describe the results enabling the computation of the identification rate. Sec.~\ref{sec:Summary} delivers concluding remarks.

\section{Related Literature}\label{sec:litSurvey}

\subsection{Directly related work}
In the current paper we focus on discrete alphabets only, following \cite{Ahlswede97}. A parallel result, with complete characterization of the identification rate (and exponent) for the Gaussian case and quadratic distortion, appears in \cite{ICW13_DCC},\cite{ICW13_IT}. The identification exponent problem was originally studied in \cite{Ahlswede97} for the variable length case, where the resulting exponent depends on an auxiliary random variable with unbounded cardinality. A bound on the cardinality has been obtained recently in \cite{IDexponent}, where the exponent for fixed-length schemes is also found (and is different than that of the variable length schemes, unlike the identification rate -- see Prop.~\ref{prop:VL=FL} below). In the special case of exact match queries (i.e. identification with $D=0$ for Hamming distance), the exponent was studied in \cite{IDexponentExact}.

\subsection{Other work in information theory}
Another closely related work is the one by Tuncel et al. \cite{tuncel2004rate}, where a similar setting of searching in a database was considered. In that work the search accuracy was addressed by a reconstruction requirement with a single-letter distortion measure that is side-information dependent (and the tradeoff between compression and accuracy is that of a Wyner-Ziv \cite{WynerZiv} type). In contrast, in the current paper the search accuracy is measured \emph{directly} by the probability of false positives.

A different line of work geared at identifying the fundamental performance limits of database retrieval includes \cite{OSullivan2002,Willems03}, which characterize the maximum rate of entries that can be reliably identified in a database. These papers were extended in \cite{Westover08,Tuncel09} allowing compression of the database, and in \cite{Tuncel12submitted} to the case where sequence reconstruction is also required. In all of these papers, the underlying assumption is that the original sequences are corrupted by noise before enrolled in the database, the query sequence is \emph{one of those original sequences}, and the objective is to identify which one.
There are two fundamental differences between this line of work and the one in the current paper.
First, in our case the query sequence is random (i.e. generated by nature) and does not need to be a sequence that has already been enrolled in the database.
Second, in our problem we are searching for sequences that are \emph{similar} to the query sequence (rather than an exact match).

\subsection{Hashing and related concepts}

The term \emph{hashing} (see, e.g. \cite{hashingBook}) generally refers to the process of representing a complex data entity with a short signature, or hash. Classically, hashing is used for quickly identifying exact matches, by simply comparing the hash of the source sequence and that of the query sequence. Hashing has been extended to detect membership in sets, a method known as the Bloom Filter \cite{Bloom1970} (with many subsequent improvements, e.g. \cite{Porat09_Opt_Bloom_Filter_matrix}). Here, however, we are interested in similarities, or ``approximate'' matches.

The extension of the hashing concept to similarity search is called Locality Sensitive Hashing (LSH), which is a framework for data structures and algorithms for finding similar items in a given set
(see \cite{AndoniI08} for a survey). LSH trades off accuracy with computational complexity and space, and false negatives are allowed. Several fundamental points are different in our setting.
First, we study the information-theoretic aspect of the problem, i.e. concentrate on space only (compression rate) and ignore computational complexity in an attempt to understand the amount of information \emph{relevant to querying} that can be stored in the short signatures.
Second, we do not allow false negatives, which, as discussed above, are inherent for LSH. Third, in the general framework of LSH (and also in hashing), the general assumption is that the data is fixed and that the hashing functions are random. This means that the performance guarantees are given as low failure probabilities, where the probability space is that of the random functions. However, for a given database, the hashing function is eventually fixed, which means that there always exist source and/or query sequences for which the scheme will \emph{always} fail. In our case, the scheme is deterministic, false negatives never occur (by design), and the probability of false positive depends on the probabilistic assumptions on the data.

Another related idea is that of dimensionality reduction techniques that preserve distances,
namely based on Johnson-Lindenstrauss type embeddings \cite{JohnsonLindenstrauss84}. Such embeddings take a set of points in space, and transform each point to a point in a lower-dimensional space, with a guarantee that the distance between these points is approximately preserved. However, note that such mappings generally depend on the elements in the database --
so that the distance preservation property cannot apply to \emph{any} query element outside the database, making the guarantee for zero false negative impossible without further assumptions. In fact, the original proof of the lemma in \cite{JohnsonLindenstrauss84} results in a guarantee for any two points in space, but this guarantee is probabilistic, and therefore cannot match our setting (similarly to LSH).

The process of compressing a sequence to produce a short signature can be also thought of as a type of \emph{sketching} (see, e.g. \cite{sketchingNotes}), which is a computational framework for succinct data representation that still allows performing different operations with the data.

\subsection{Practical examples of compression for similarity identification}
The idea of using compression for accelerating similarity search in databases is not new. Earlier practical examples include the \emph{VA-file} scheme \cite{VAfile}, which uses scalar quantization of each coordinate of the source sequence in order to form the signature. The VA-file approach demonstrates that compression-based similarity search systems can outperform tree-based systems for similarity search, providing further motivation to study the fundamental tradeoff between compression and search accuracy. The VA-file scheme has been generalized to \emph{vector} quantization in \cite{VQfile}, showing further improvements in both computational time and number of disk I/O operations.

In the machine learning literature, the term `semantic hashing' \cite{semantichashing} refers to a transformation that maps similar elements to bit strings that have low Hamming distance. Extensions of this concept include \cite{spectralhashing,ragiskyHashing}. We comment that in neither of these papers there is a guarantee for zero false negatives, as in the setting considered in the current paper.

We emphasize again that the results in the current paper are concerned with the amount of compression only, and ignore the computational complexity (as is typical for information theoretical results). Nevertheless, the fundamental limits, such as $R_\ID(D)$, characterize the playing field at which practical schemes should be evaluated.

\section{Problem Formulation}\label{sec:setting}
\subsection{Notation}
Throughout this paper, boldface notation $\bx$ denotes a column vector of elements $[x_1,...x_n]^T$. Capital letters denote random variables (e.g. $X,Y$), and $\bX,\bY$ denote random vectors. We use calligraphic fonts (e.g. $\cX,\cY$) to represent the finite alphabets. $\log(\cdot)$ denotes the base-$2$ logarithm, while $\ln(\cdot)$ is used for the usual natural logarithm.

We measure the similarity between symbols with an arbitrary per-letter distortion measure $\rho:\cX \times \cY \ra \Reals_+$. For length $n$ vectors, the distortion is given by
\begin{equation}
  d(\bx,\by) \triangleq \frac{1}{n} \sum_{i=1}^n \rho(x_i,y_i).
\end{equation}
We say that $\bx$ and $\by$ are $D$-\emph{similar} when $d(\bx,\by)\leq D$, or simply \emph{similar} when $D$ is clear from the context.

\subsection{Identification Systems}
A rate-$R$ identification system $(T,g)$ consists of a \emph{signature assignment}
\begin{align}
T : \cX^n \ra \{1,2,\dots,2^{nR}\}
\end{align}
and a \emph{decision function}
\begin{align}
g : \{1,2,\dots,2^{nR}\}\times \cY^n \ra \{\NO, \MAYBE\}.
\end{align}

A system $(T,g)$ is said to be $D$-\emph{admissible}, if for any $\bx,\by$ satisfying $d(\bx,\by)\leq D$, we have
\begin{equation}\label{eqn:maybe}
  g(T(\bx),\by) = \MAYBE.
\end{equation}
This notion of $D$-{admissibility} motivates the use of ``$\NO$" and ``$\MAYBE$" in describing the output of $g$:
\begin{itemize}
\item If $g(T(\bx),\by) = \NO$, then $\bx$ and $\by$ can not be $D$-similar.
\item If $g(T(\bx),\by) = \MAYBE$, then $\bx$ and $\by$ are possibly $D$-similar.
\end{itemize}
Stated another way, a $D$-{admissible} system $(T,g)$ does not produce false negatives. Thus, a natural figure of merit for a $D$-{admissible} system $(T,g)$ is the frequency at which false positives occur (i.e., where $g(T(\bx),\by) = \MAYBE$ and $d(\bx,\by)>D$). To this end, let $P_X$ and $P_Y$ be probability distributions on $\cX,\cY$ respectively, and assume
that the vectors $\bX$ and $\bY$ are independent of each other and drawn i.i.d.\ according to $P_X$ and $P_Y$ respectively.  Define the \emph{false positive event}
\begin{align}
\mathcal{E} = \{ g(T(\bX),\bY) = \MAYBE,  d(\bX,\bY) > D\}, \label{FPeventDefn}
\end{align}
and note that, for any $D$-admissible system $(T,g)$, we have
\begin{align}
\Pr \{ g(T(\bX),\bY) = \MAYBE \}
&= \Pr \{ g(T(\bX),\bY) = \MAYBE | d(\bX,\bY)\leq D\}\Pr\{ d(\bX,\bY)\leq D\} \notag\\
&\quad+ \Pr\{ g(T(\bX),\bY) = \MAYBE,  d(\bX,\bY) > D\} \label{eqn:FP_maybeRelationPre}\\
&=\Pr\{ d(\bX,\bY)\leq D\} + \Pr\{\mathcal{E}\}, \label{eqn:FP_maybeRelation}
\end{align}

where \eqref{eqn:FP_maybeRelation} follows since $\Pr \{ g(T(\bX),\bY) = \MAYBE | d(\bX,\bY)\leq D\}=1$ by $D$-admissibility of $(T,g)$.  Since $\Pr\{ d(\bX,\bY)\leq D\}$ does not depend on what scheme is employed, minimizing the false positive probability $\Pr\{\mathcal{E}\}$ over all $D$-admissible schemes $(T,g)$ is equivalent to minimizing $\Pr \{ g(T(\bX),\bY) = \MAYBE \}$. Also note, that the only interesting case is when $\Pr\{d(\bX,\bY)\leq D \} \ra 0$ as $n$ grows, since otherwise almost all the sequences in the database will be similar to the query sequence, making the problem degenerate (since almost all the database needs to be retrieved, regardless of the compression). In this case, it is easy to see from \eqref{eqn:FP_maybeRelationPre} that $\Pr\{\mathcal{E}\}$ vanishes if and only if the conditional probability
\begin{equation}
  \Pr \{ g(T(\bX),\bY) = \MAYBE | d(\bX,\bY)> D\}
\end{equation}
vanishes as well.
In view of the above, we henceforth restrict our attention to the behavior of $\Pr \{ g(T(\bX),\bY) = \MAYBE \}$.  In particular, we study the tradeoff between the rate $R$ and $\Pr \{ g(T(\bX),\bY) = \MAYBE \}$.
This motivates the following definitions:
\begin{defn}\label{def:achRate}
    For given distributions $P_X, P_Y$ and a similarity threshold $D$, a rate $R$ is said to be $D$-\emph{achievable} if there exists a sequence of admissible schemes $(T^{(n)},g^{(n)})$ with rates at most $R$, satisfying
    \begin{equation}
      \lim_{n\ra\infty} \Pr\left\{g^{(n)}\left(T^{(n)}(\bX),\bY \right) = \MAYBE\right\} = 0. \label{eqn:reliable}
    \end{equation}
\end{defn}
\begin{defn}\label{def:RID}
    For given distributions $P_X, P_Y$ and a similarity threshold $D$, the \emph{identification rate}  $R_\ID(D,P_X,P_Y)$ is the infimum of $D$-achievable rates.  That is,
   \begin{align}
     R_\ID(D) \triangleq \inf \{ R : R~ \mbox{is $D$-achievable}\},
   \end{align}
   where an infimum over the empty set is equal to $\infty$.
\end{defn}

It is not hard to see that $R_\ID(D)$ must be nondecreasing. To see this, note that any sequence of schemes at rate $R$ that achieve vanishing probability of $\MAYBE$ for similarity threshold $D$, is also admissible for any threshold $D'\leq D$, so if $R$ is $D$-achievable, then it is also $D'$-achievable. In other words, a higher similarity threshold is a more difficult task (i.e. requires higher compression rate). Therefore, analogously to the definition of $R_\ID(D)$, we define $D_\ID(R)$ as the maximal achievable similarity threshold for fixed rate schemes.

\medskip
The definitions of an achievable rate and the identification rate are in the same spirit of the rate distortion function (the rate above which a vanishing probability for excess distortion is achievable),
and also in the spirit of the channel capacity (the rate below which a vanishing probability of error can be obtained). See, for example, Gallager~\cite{GallagerInfoTheoryBook}.

\subsection{Variable Length Identification Systems}
In \cite{Ahlswede97}, the authors study a similar setting, where the compression is of \emph{variable length}. In that spirit, we define the corresponding  variable-length quantities:

A \emph{variable length} identification system $(T_{vl},g_{vl})$ consists of a signature assignment
\begin{align}
T_{vl} : \cX^n \ra B,
\end{align}
where $B\subseteq \{0,1\}^*$ is a prefix-free set,
and a decision function
\begin{align}
g_{vl} : B\times \cY^n \ra \{\NO, \MAYBE\}.
\end{align}
The rate of the system is given by
\begin{equation}
  R = \frac{1}{n}\EE\left[\mbox{length}\left(T_{vl}(\bX)\right)\right].
\end{equation}
As before, a scheme is said to be admissible, if for any $\bx,\by$ satisfying $d(\bx,\by)\leq D$, we have
\begin{equation}\label{eqn:maybeVL}
  g_{vl}(T_{vl}(\bx),\by) = \MAYBE.
\end{equation}

Analogously to Definitions~\ref{def:achRate} and \ref{def:RID}, we define the variable-length identification rate, denoted $R_\ID^{vl}$ as the infimum of achievable rates for variable length identification systems. Clearly, any rate $R$ that is achievable with fixed-length schemes is also achievable with variable length schemes, and therefore $R_\ID \geq R_\ID^{vl}$. It turns out that both quantities are actually equal:
\begin{prop}\label{prop:VL=FL}
    The identification rate for variable rate is the same as that for fixed rate, i.e.
    \begin{equation}
      R_\ID(D) = R_\ID^{vl}(D)
    \end{equation}
\end{prop}
The proof of the proposition, given in detail in Appendix~\ref{app:VL=FL}, is based on a simple meta-argument, that essentially says that any variable length scheme can be used as a building block to construct a fixed-length scheme. The argument is based on the concatenation of several input sequences into a larger one, and then applying a variable length scheme to each of the sequences. This will result in a variable length scheme, but with high probability, most of the signatures will have overall length bounded by some fixed length, enabling the conversion to a fixed-length scheme.

There are two direct consequences of Prop.~\ref{prop:VL=FL} that will enable the evaluation of $R_\ID(D)$: In order to show that a given rate is achievable with fixed-rate schemes, it is possible to consider variable length schemes, as in \cite{Ahlswede97}. On the other hand, in order to prove a converse, it suffices to consider fixed-length schemes, slightly simplifying the proof. This is the path we take in the paper.

\section{Main Results}\label{sec:MainResults}

\subsection{The Identification Rate}
Define the following distance between distributions $P_X,P_Y$:
\begin{equation}
   \bar\rho(P_X,P_Y) \triangleq \min\EE[\rho(X,Y)],
\end{equation}
where the minimization is w.r.t. all random variables $X,Y$ with marginal distributions $P_X$ and $P_Y$, respectively. This distance goes by many names, such as the Wasserstein (or Vasershtein) distance, the Kantorovich distance and also the Transport distance (see \cite{gray1975rhoBar}, and also \cite{villani2009optimal} for a survey).

Define the (informational) identification rate as
\begin{align}\label{eqn:R_ID_info}
    \bar R_\ID(D) &=\min_{P_{U|X}: \sum_{u\in\cU} P_U(u) \bar\rho(P_{X|U}(\cdot|u),P_Y)\geq D} I(X;U),
\end{align}
where $U$ is any random variable with finite\footnote{The cardinality of $\cU$ can be taken as $|\cX|+2$, according to \cite[Lemma 3]{Ahlswede97}. However, in the sequel we improve the cardinality bound, see Subsection~\ref{ssec:comp}.} alphabet $\cU$, that is independent of $Y$.

It follows from \cite[Theorem 2]{Ahlswede97} that when $R>\bar R_\ID(D)$, there exist (variable-length) identification schemes with FP probability that vanishes exponentially with $n$ (the explicit connection to the limiting rate $\bar R_\ID(D)$ is made in \cite[Eq. (2.21)]{Ahlswede97}). This fact, combined with Prop.~\ref{prop:VL=FL} above implies that $R_\ID(D) \leq \bar R_\ID(D)$. However, it remains open whether $R_\ID(D)$ is even strictly positive. In the related case studied in \cite{Ahlswede97}, where the probability of both FN and FP events are required to vanish, it was shown \cite[Thm. 1]{Ahlswede97} that the achievable rate in this sense is equal to zero, so according to \cite{Ahlswede97}, ``the only problem left to investigate is the case tradeoff between the rate $R$ and the two error exponents [of the FP and FN events]''. Our first result below shows that the restriction to the case of no FN (also called a `one-sided error') is, in fact, very interesting.

\begin{theorem}[The Identification Rate Theorem]\label{thm:RID}
    \begin{equation}
      R_\ID(D) = \bar R_\ID(D),
    \end{equation}
    i.e. the identification rate is given in \eqref{eqn:R_ID_info}. Moreover, if $R < R_\ID(D)$, then the probability of $\MAYBE$ converges to $1$ exponentially fast.
\end{theorem}

A few comments are in order regarding the theorem, whose proof is given in detail in Sec.~\ref{sec:ProofConverse}:
\begin{itemize}
  \item Theorem~\ref{thm:RID} states that the case where the probability of FN events is \emph{equal} to zero is inherently different from the case first discussed in \cite[Thm. 1]{Ahlswede97}, where the FN probability is only required to vanish with $n$: here the rate problem is not degenerate (since the minimal achievable rate question does not always give zero), and is in fact, more along the lines of the classical results in information theory, such as channel capacity and classical rate distortion compression.
  \item As discussed above, the direct part of the theorem follows from \cite{Ahlswede97} and Prop.~\ref{prop:VL=FL}. Nevertheless, in Sec.~\ref{sec:ProofConverse} we shall also outline how to prove the achievability part directly, based on a version of the type covering lemma.
  \item The techniques used in \cite{Ahlswede97} for proving a converse result on the error exponents are not strong enough for proving the converse for Thm.~\ref{thm:RID}. In the proof here, we utilize some of the tools developed in \cite{Ahlswede97}, namely the \emph{inherently typical subset lemma}, and augment them with the \emph{blowing-up lemma} (\cite{marton1996bounding}, see also \cite[Lemma 12]{RaginskySasonConcentration}). The purpose of the blowing up lemma in this context is to take an event whose probability is exponentially small, but with a very small exponent, and transform it to a related event, whose probability is very close to 1.  See Sec.~\ref{sec:ProofConverse} for details.
\end{itemize}

\subsection{Special Case: Distortion Measures Satisfying the Triangle Inequality}

Consider the special case where $\cX = \cY$, and for all $x,y,z \in \cX$, we have
\begin{equation}\label{eqn:triangleInequality}
  \rho(x,y) + \rho(y,z) \geq \rho(x,y).
\end{equation}
In other words, the distortion measure satisfies the \emph{triangle inequality}, which is a common property of distance / similarity measures. For simplicity, also assume that the measure is symmetric, i.e. that for all $x,y\in\cX$, $\rho(x,y)= \rho(y,x)$ (later on we discuss the generalization to the non-symmetric case). In this case there are intuitive compression schemes that naturally come to mind. The main idea is to use the compressed representation for reconstructing an approximation $\hat\bx$ of the source, and then to use this reconstruction to decide whether $\bx$ and $\by$ are similar or not. For example, suppose that the source was indeed reconstructed as $\hat\bx$, and also assume that we know the value of $d(\bx,\hat\bx)$ (adding this value to the compressed representation of $\hat\bx$ is a negligible addition to the rate). Consider the following decision rule, based on the pair $(\hat\bx,d(\bx,\hat\bx))$:
\begin{equation}\label{eqn:triangleDecisionRule}
  g((\hat\bx,d(\bx,\hat\bx)),\by)= \left\{
  \begin{array}{ll}
    \NO, & \hbox{$d(\hat\bx,\by) > d(\bx,\hat\bx) + D$;} \\
    \MAYBE, & \hbox{otherwise.}
  \end{array}
\right.
\end{equation}
This rule is admissible because in cases where $\bx$ and $\by$ are $D$-similar, it follows by the triangle inequality that
\begin{align}
    d(\hat\bx,\by)
    &\leq d(\hat\bx,\bx) + d(\bx,\by)\\
    &\leq d(\hat\bx,\bx) + D,
\end{align}
resulting in the decision function in \eqref{eqn:triangleDecisionRule} returning a $\MAYBE$. The process is illustrated in Fig.~\ref{fig:triangle}.

\begin{figure}
  \centering
  \includegraphics[width=4in]{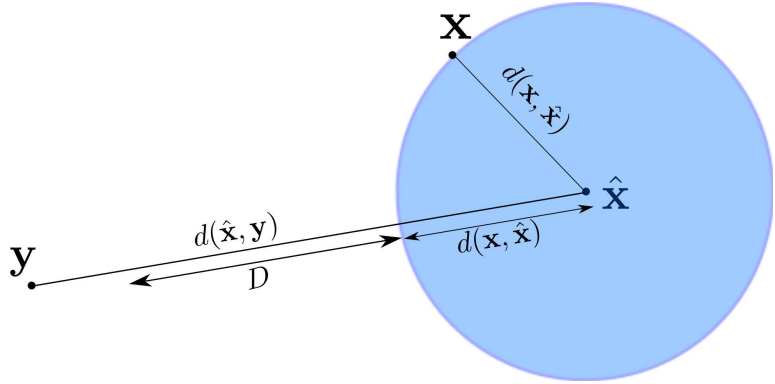}\\
  \caption{Illustration of the triangle-inequality decision rule. If $d(\hat\bx,\by) > D + d(\bx,\hat\bx)$, then it is certain that $d(\bx,\by) > D$, and we can therefore safely declare that $\bx$ and $\by$ are not similar.}\label{fig:triangle}
\end{figure}

The remaining question is, then, how to choose the lossy representation $\hat\bx$, and whether this results in an optimal scheme (i.e. whether the optimal $R_\ID(D)$ is achieved). We survey two schemes based on the triangle inequality principle, and discuss the compression rate achievable with each scheme. In general, neither can be shown to achieve $R_\ID(D)$.

\smallskip

\underline{The naive scheme: \LCT\ (Lossy Compression signatures and triangle ineq. decision rule)}

In this scheme, we use standard lossy compression in order to represent $\hat\bx$ with fixed rate $R$, i.e. we optimize for a reconstruction that minimizes $d(\bx,\hat\bx)$.
When the compression rate is $R$, it is known that the attained distortion $d(\bx,\hat\bx)$ is close, with high probability and for long enough sequences, to the distortion-rate function $D(R)$.

The full details of the $\LCT$ scheme are given in Section~\ref{sec:Triangle}, along with the proof of the following theorem, which characterizes the similarity threshold supported by any given compression rate:
\begin{theorem}\label{thm:Naive}
    Any similarity threshold below $D^\LCT_\ID(R)$ can be attained with a \LCT\ compression scheme of rate $R$, where
    \begin{equation}\label{eqn:DIDLCT}
      D^\LCT_\ID(R) \triangleq \EE[\rho(\hat X,Y)]-D(R),
    \end{equation}
    where $D(R)$ is the classical distortion-rate function of the source $X$, and $\hat X$, which is independent of $Y$, is distributed according to the marginal distribution of the $D(R)$ achieving distribution (if there are several $D(R)$-achieving distributions, take one that maximizes $\EE[\rho(\hat X,Y)]$).
\end{theorem}
\smallskip

We denote by $R_\ID^\LCT(D)$ the inverse function of $D^\LCT_\ID(R)$, i.e. the compression rate that is achieved for a similarity threshold $D$. This scheme is suboptimal in general, i.e. there are cases for which $R^\LCT_ID(D) > R_\ID(D)$, but in some cases they are equal. For example, the symmetric binary source with Hamming distortion is one case in which this naive scheme is optimal. This case is discussed in detail in \cite{IdoiaAllerton2013}, where an actual scheme is implemented based on lossy compression algorithms.

\bigskip

\underline{An improved scheme: \TCT\ (``Type Covering'' signatures and triangle ineq. decision rule)}

The expression in \eqref{eqn:DIDLCT} gives rise to the following intuitive idea: in the distortion rate case, we wish to minimize the distortion, with a constraint on the mutual information (that controls the compression rate). The free variable in the optimization is the transition probability $P_{\hat X|X}(\hat x | x)$. So far this is in agreement with \eqref{eqn:DIDLCT}, as we wish that the similarity threshold will be maximized. However, the expectation term in \eqref{eqn:DIDLCT} also depends on the transition probability $P_{\hat X|X}(\hat x | x)$, raising the question whether both terms should be optimized together. The answer to this question is positive. The key step is to use a general type covering lemma for generating $\hat\bx$ (using a distribution that does not necessarily minimize the distortion between $X$ and $\hat X$). This idea is made concrete in the following theorem (whose proof is given in section ~\ref{sec:Triangle}):
\begin{theorem}\label{thm:LessNaive}
    Any similarity threshold below $D^\TCT_\ID(R)$ can be attained by a \TCT\ compression scheme of rate $R$, where
    \begin{equation}\label{eqn:DIDTCT}
      D^\TCT_\ID(R) \triangleq \max_{P_{\hat X|X} : I(X;\hat X) \leq R} \EE[\rho(\hat X,Y)]-\EE[\rho(X,\hat X)],
    \end{equation}
where on the RHS, $\hat X$ and $Y$ are independent.
\end{theorem}

We denote by $R_\ID^\TCT(D)$ the inverse function of $D^\TCT_\ID(R)$, i.e. the compression rate that is achieved by a \TCT\ scheme for a similarity threshold $D$. It can be also written as \begin{equation}
  R_\ID^\TCT(D) = \min_{P_{\hat X|X}:
  \EE[\rho(\hat X,Y)]-\EE[\rho(X,\hat X)] \geq D  } I(X;\hat X).
\end{equation}
We also note that the \TCT\ scheme is a natural extension of the scheme given in \cite[Theorem~3]{ICW13_IT}, which applies for continuous sources and quadratic distortion.

\bigskip

It is not hard to see that $R_\ID^\TCT(D) \leq R_\ID^\LCT(D)$, since the distortion-rate achieving distribution in \eqref{eqn:DIDLCT} is a feasible transition probability for the expression in \eqref{eqn:DIDTCT}. However, the \TCT\ scheme is not optimal in general, as we shall see later on.
\bigskip

So far we have, in general, that
\begin{equation}
      R_\ID(D) \leq R_\ID^\TCT(D) \leq R_\ID^\LCT(D).
\end{equation}
There are special cases, however, where some of the above inequalities are actually equalities (proved in Sec.~\ref{sec:Triangle}):
\begin{itemize}
  \item For the binary-Hamming case, we have
  \begin{equation}\label{eqn:special1}
    R_\ID(D) = R_\ID^\TCT(D).
  \end{equation}
  \item If $\sum_{y\in\cY} P_Y(y) \rho(\hat x,y)$ is constant for all $\hat x \in \cX$, then
  \begin{equation}\label{eqn:special2}
      R_\ID^\TCT(D) = R_\ID^\LCT(D) = R(D_0-D),
  \end{equation}
  where $R(\cdot)$ is the standard rate-distortion function and $D_0 \triangleq \sum_{y\in\cY} P_Y(y) \rho(\hat x,y)$.
    \item As a consequence, in the binary-Hamming case where $Y$ is symmetric, both \eqref{eqn:special1} and \eqref{eqn:special2} hold, and we have
    \begin{equation}
        R_\ID(D) = R_\ID^\TCT(D) = R_\ID^\LCT(D) = R(\tfrac{1}{2}-D),
    \end{equation}
    where $R(\cdot)$ is the rate distortion function for the source $X$ and Hamming distortion.
\end{itemize}

\bigskip
Next, we provide an easily computable \emph{lower bound} on the identification rate, that holds for the case of Hamming distortion:
\begin{theorem}\label{thm:UniversalLB_Hamming}
    For Hamming distortion, the identification rate is always lower bounded by
\begin{equation}
  R_\ID(D) \geq \left[2D^2 \log e - D(P_X||P_Y)\right]^+.
\end{equation}
In particular, when $P_X = P_Y$, we have
\begin{equation}
  R_\ID(D) \geq 2D^2 \log e.
\end{equation}
\end{theorem}
\begin{proof}Appendix~\ref{app:UniversalLB_Hamming}\end{proof}

\bigskip

In Fig.~\ref{fig:ternary} we plot the three rates: $R_\ID^\LCT(D)$, $R_\ID^\TCT(D)$ and $R_\ID(D)$ for the case of $\cX = \cY = \{0,1,2\}$, $P_X = P_Y = [.8 .1 .1]$ and Hamming distortion. As seen in the figure, the three rates are different, indicating that neither of the \LCT\ and \TCT\ schemes is  optimal. We also plot the lower bound from Theorem~\ref{thm:UniversalLB_Hamming}.
\begin{figure}
  \centering
  \includegraphics[width=6in]{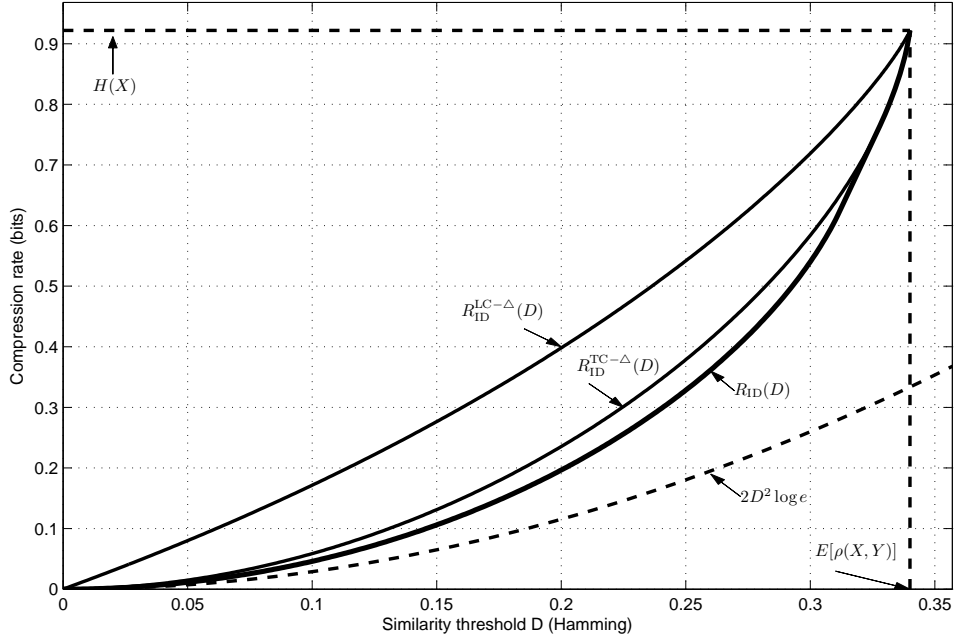}\\
  \caption{Three compression rates for the ternary source $P_X = [.8,.1 .1]$. The distribution $P_Y$ of the query sequence is the same as $P_X$. The distortion measure is Hamming.}\label{fig:ternary}
\end{figure}

\subsection{Computing the Identification Rate}\label{ssec:comp}

The identification rate $R_\ID(D)$ is given in \eqref{eqn:R_ID_info} as an optimization program of  single-letter information-theoretic quantities, with an auxiliary random variable with alphabet $\cU$. In order to actually compute the value of $R_\ID(D)$, one has to (a) obtain a bound on the cardinality of $\cU$, and (b) be able to efficiently solve the optimization program, which is non-convex.

In order to tackle the cardinality issue, let us define
\begin{equation}
  R_\ID^k(D) \triangleq \min_{P_{U|X}: \sum_{u\in\cU} P_U(u) \bar\rho(P_{X|U}(\cdot|u),P_Y)\geq D} I(X;U),
\end{equation}
where the alphabet $\cU$ is given by $k$. Clearly, $R_\ID^k(D)$ is a decreasing function of $k$ (since a lower value of $k$ can be considered a special case of the higher $k$). In \cite[Lemma 3]{Ahlswede97}, the standard \emph{support lemma} (see, e.g. \cite[Lemma 15.4]{CsiszarKorner2nd}) is used to show that $R_\ID(D) = R^{|\cX|+2}_\ID(D)$. In other words, taking $|\cU|=|\cX|+2$ suffices in order to obtain the true value of $R_\ID(D)$. We improve this result as follows:
\begin{theorem}\label{thm:cardinality}
    For any $D$, we have
    \begin{equation}
      R_\ID(D) = R^{|\cX|+1}_\ID(D).
    \end{equation}
    Furthermore, the entire curve $R_\ID(D)$ can be obtained by calculating the curve $R^{|\cX|}_\ID(D)$ for all $D$, and then taking the lower convex envelope. 
\end{theorem}

Remarks:
\begin{itemize}
  \item The proof, given in detail in Sec.~\ref{sec:computation}, is based on the support lemma, but uses it in a more refined way than in \cite[Lemma 3]{Ahlswede97}, in a manner similar in spirit to \cite{Jana2009Cardinality}.
  \item The second part of the result follows from the fact that whenever $(R_\ID(D),D)$ is an \emph{exposed} point of the convex region of achievable pairs,
    \begin{equation}\label{eqn:CardinalityImproved}
       R_\ID(D) = R^{|\cX|}_\ID(D).
    \end{equation}
    See Sec.~\ref{sec:computation} for details.
  \item Taking the lower convex envelope is also necessary. In other words, we sometimes have a strict inequality of the form $R^{|\cX|+1}_\ID(D) < R^{|\cX|}_\ID(D)$. For example, in the case of ternary alphabet and Hamming distortion as in Fig.~\ref{fig:ternary} above, we have such a strict inequality for $D=0.32$.
\end{itemize}

\bigskip

The harder problem is the non-convexity of the optimization in \eqref{eqn:R_ID_info} (not to be confused with the fact that the function $R_\ID(D)$ itself is a convex function of $D$, see \cite[Lemma 3]{Ahlswede97}). While the target function (the mutual information) \emph{is} convex, the feasibility region is not, which makes the optimization hard. In order to tackle this issue, we show that this region is the complementary of a convex polytope. Then, we propose a method for reducing the optimization program \eqref{eqn:R_ID_info} to a sequence of \emph{convex} optimization programs that can be solved easily (e.g. via \texttt{cvx} \cite{cvx}), and the minimum among the solutions of those programs is equal to $R_\ID(D)$.

To illustrate this idea, consider the optimization of a convex function $f:\Reals^n \ra \Reals$ over the set $\Reals^n \setminus \Xi$, where $\Xi$ is a convex polytope:
\begin{equation}
  \inf_{\bz \in \Reals^n \setminus \Xi} f(\bz).
\end{equation}
Any polytope can be written as
\begin{equation}
  \Xi = \{\bz \in \Reals^n : \ba_i^T\bz \leq b_i \mbox{ for all } 1 \leq i \leq m\},
\end{equation}
where $\{\ba_i\}$ are $m$ length-$n$ vectors and $\{b_i\}$ are $m$ scalars, where $m$ corresponds to the number of facets of the polytope $\Xi$. Rewriting the optimization program gives
\begin{align}
  \inf_{\bz \in \Reals^n \setminus \Xi} f(\bz)
  & = \min_{\bz \in \Reals^n: \exists_{i}:\ba_i^T \bz \geq b_i} f(\bz) \\
  & = \min_{1 \leq i \leq m } \ \min_{\bz \in \Reals^n: \ba_i^T \bz \geq b_i} f(\bz).
\end{align}
Each of the new optimization programs is a minimization of a the original convex function, but now with \emph{linear} constraints, and therefore can be calculated easily.

While this method does not scale well with the alphabet size (since the number of facets of the polytope grows very quickly with $|\cX|$), it still provides a method for calculating $R_\ID(D)$ for small values of $|\cX|$. For example, the $R_\ID(D)$ curve in Fig.~\ref{fig:ternary} was obtained with this method. The full details of the reduction method, along with simplifications for the Hamming distortion case, are given in Sec.~\ref{sec:computation}.

\section{Proof of the Identification Rate Theorem}\label{sec:ProofConverse}
In this section we prove Theorem~\ref{thm:RID}. After giving a high-level overview of the proof, we introduce additional notation and review the basic tools that are used in the proof. The proof itself is given in Subsection \ref{ssec:proofRIDthm}.

\subsection{Proof roadmap}
We start with an informal overview of the proof. First, note that it suffices to consider only \emph{typical} sequences $\bx$, i.e. sequences whose (first-order) empirical distribution is close to the true one $P_X$ (denoted $\typecl_{P_X}$, see a formal definition below). The same holds for the query sequences $\by$.

The achievability scheme is based on constructing a code, which is a set of sequences from another alphabet $\cU^n$, that ``covers'' the typical set $\typecl_{P_X}$. The covering is in the sense that for each $\bx\in\typecl_{P_X}$, there exists a word $\bu$ in the code, s.t. $\bx,\bu$ will have a first-order empirical distribution that is close to some given joint distribution (which is given as a design parameter $P_{U|X}$, along with the size of the alphabet $\cU$). Such a covering is guaranteed to exist by a version of the type covering lemma (stated below), and the code rate is given by the mutual information between $X$ and $U$, which define the joint distribution. The signature of a sequence $\bx$, $T(\bx)$, is defined as the index to the sequence $\bu$ in the code that covers $\bx$. The decision process $g(\cdot,\cdot)$ simply declares $\MAYBE$ (given $\bu$ and $\by$), if there exists a typical sequence $\bx$ that is mapped to $\bu$, that is also similar to $\by$. The scheme is admissible, and all that remains is to evaluate the probability of $\MAYBE$, i.e. the probability that $\bY$ falls in the $D$-`expansion' of the set of sequences $\bx$ that are mapped to $\bu$. See Fig.~\ref{fig:covering} for an illustration.

It can be shown that if the similarity threshold $D$ satisfies
\begin{equation}\label{eqn:blabla}
  D < \sum_u P_U(x) \bar\rho(P_{X|U}(\cdot|u),P_Y),
\end{equation}
the fraction of sequences for which a $\MAYBE$ is declared vanishes with $n$. This leads to the achievablity of the rate $\bar R_\ID(D)$.

\begin{figure}
  \centering
  \includegraphics[width=6in]{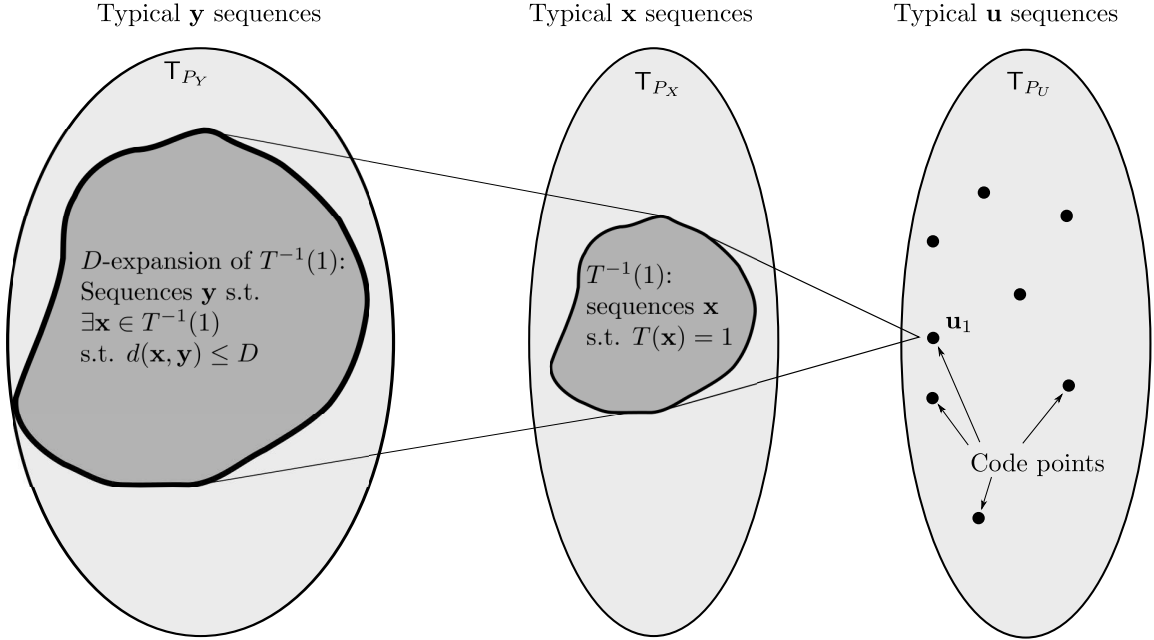}\\
  \caption{Illustration of the achievability scheme. Each sequence in the typical set $\typecl_{P_X}$ is mapped to (the index of) a single code point $\bu$. A $\MAYBE$ is declared if $\by$ falls in the $D$-expansion of the set of $\bx$-sequences that are mapped to $\bu$.}\label{fig:covering}
\end{figure}

For the converse part, we follow the steps of \cite{Ahlswede97} and essentially show that any compression scheme performs as well as a scheme that was considered in the achievability part. In particular, we show that there exists a distribution $P_{U|X}$ for which each of the sets of $\bx$-sequences mapped to the same $i$, contains a set which is called `inherently typical' w.r.t. $P_{U_X}$ (see below). While in the achievability part we claimed that if \eqref{eqn:blabla} holds then the probability of $\MAYBE$ vanishes, here we need to claim the opposite, i.e. that if \eqref{eqn:blabla} does not hold, then the probability of $\MAYBE$ cannot go to zero. We note that in \cite{Ahlswede97} the argument can only lead to a claim of the sort ``if \eqref{eqn:blabla} does not hold, then the probability of $\MAYBE$ cannot vanish \emph{exponentially} with $n$''. Here, however, we require a stronger result. In order to proceed, we use the blowing-up lemma (see below), and show that if \eqref{eqn:blabla} does not hold, then the probability of $\MAYBE$ \emph{converges to 1}, and this convergence is \emph{exponential} in $n$. This can be regarded as a `exponentially strong converse' (see, e.g. \cite{WangIngberKochmanJSCCstrongConverse} for an overview of converse types).

\subsection{Additional Notation}

We shall use the method of types \cite{CsiszarKorner2nd}. Let $\cP(\cX)$ denote the set of all probability distributions on the alphabet $\cX$. We denote by $\cP(\cX \ra \cY)$ the set of all stochastic matrices (or, equivalently, conditional distributions, or channels) from alphabet $\cX$ to $\cY$. On occasions, we deal with vectors of different lengths. For this purpose, we use the notation $x^k$ as short hand for the vector $[x_1,...,x_k]$, so, for example, $x^n$ and $\bx$ shall denote the same thing. For a sequence $\bx \in \cX^n$ and $a\in\cX$, let $N(a|\bx)$ denote the number of occurrences of $a$ in $\bx$. The \emph{type} of the sequence $\bx$, denoted $\type_\bx$, is defined as the vector $\frac{1}{n}[N(1|\bx),N(2|\bx),...,N(|\cX|\ |\bx)]$. For any sequence length $n$, let $\cP_n(\cX)$ denote the set of all possible types of sequences of length $n$, (also called $n$-types):
\begin{equation}
    \cP_n(\cX) \triangleq \left\{P\in\cP(\cX) |\forall x\in\cX,\ nP(x) \in \Integers_+\right\}.
\end{equation}
For a type $P \in \cP_n(\cX)$, the \emph{type class} $\typecl_P$ is defined as the set of all sequences $\bx\in\cX^n$ with type $P$ (or, equivalently, that $\type_\bx = P$). More generally, for a distribution $P \in \cP(\cX)$ (not necessarily an $n$-type), and a constant $\gamma>0$, the set of typical sequences $\typecl_{P,\gamma}$ is defined as the set of all sequences $\bx\in\cX^n$ for which:
\begin{enumerate}
  \item $|P(x) - \type_\bx(x)| \leq \gamma \quad \forall x \in \cX,$
  \item $\type_\bx(x) = 0$  whenever   $P(x) = 0$.
\end{enumerate}
If $X$ is a random variable distributed according to $P$, we shall sometimes write $\typecl_{X,\gamma}$ for $\typecl_{P,\gamma}$.

Similarly, for $V \in \cP(\cX \ra \cY)$ and $\bx \in \cX^n$, we denote by $\typecl_{V,\gamma}(\bx)$ the set of conditionally typical sequences, i.e. that
\begin{enumerate}
  \item $\left|N(x',y'|\bx,\by) - N(x'|\bx)V(y'|x')\right| \leq n\cdot \gamma \quad \forall x' \in \cX,y' \in \cY,$
  \item $N(x',y'|\bx,\by) = 0$  whenever   $V(y|x) = 0$.
\end{enumerate}
For random variables $X,Y$ where $Y$ is the output of the channel $V$ with input $X$, we'll sometimes use the notation $\typecl_{Y|X,\gamma}(\bx)$ for $\typecl_{V,\gamma}(\bx)$.

\medskip
Let $\delta_n$ be defined according to the delta convention \cite{CsiszarKorner2nd}; i.e. that $\delta_n \ra 0$, but also that $n \delta_n^2 \ra \infty$ (e.g. $\delta_n \triangleq n^{-\alpha}$, with some $\alpha \in (0,0.5)$). With this convention, we have (see \cite[Lemma 2.12]{CsiszarKorner2nd}):
\begin{itemize}
  \item If $\bX$ is distributed i.i.d. according to a distribution $P$, then
  \begin{equation}\label{eqn:probTypicalSet}
    \Pr\{\bX \in \typecl_{P,\delta_n}\} \geq 1 - \frac{|\cX|}{4n\delta_n^2}.
  \end{equation}
  \item If $\bY$ is the output of the DMC $V$ with input $\bx \in \cX^n$, then
  \begin{equation}\label{eqn:probConditionalTypicalSet}
    \Pr\{\bY \in \typecl_{V,\delta_n}(\bx)\} \geq 1 - \frac{|\cX||\cY|}{4n\delta_n^2}.
  \end{equation}
\end{itemize}

\medskip
Recall that we have defined an arbitrary distortion measure $\rho:\cX\times \cY \ra \Reals^+$. The maximal possible distortion is denoted by $\rho_{\max}$:
\begin{equation}
  \rho_{\max} \triangleq \max_{x \in \cX, y \in \cY} \rho(x,y).
\end{equation}

For a set $A \subseteq \cX^n$, we denote its $D$-expansion by
\begin{equation}
  \Gamma^D(A) \triangleq \left\{\by \in \cY^n : \exists \bx \in A \mbox{ s.t. } d(\bx,\by) \leq D\right\}.
\end{equation}

When we use the Hamming distance as the distortion measure, we denote the expansion by $\Gamma_H^D(A)$.

\subsection{A covering lemma}
The main building block for the achievability proofs in the paper is the following version of the covering lemma (see e.g. \cite[Prop. 1]{DemboWeissman2003}):
\begin{lem}[A covering lemma]\label{lem:covering}
    For any distribution $P_X\in \cP(\cX)$ and any channel $V \in \cP(\cX\ra\cU)$, there exists a mapping $\omega:\typecl_{P_X,\delta_n}\ra \cU^n$ s.t.
    \begin{equation}
      \omega(\bx) \in \typecl_{V,\delta_n}(\bx) \mbox{ for all } \bx \in \typecl_{P_X,\delta_n},
    \end{equation}
    and
    \begin{equation}
      \left|\left\{\omega(\bx) : \bx \in \typecl_{P_X,\delta_n} \right\}\right| \leq 2^{n(I(P_X,V)+\eps_n)},
    \end{equation}
    where $\eps_n = O\left(\delta_n\log(1/\delta_n)\right)$, with constants that depend only on $|\cX|$ and $|\cU|$.
\end{lem}
\medskip

Essentially, this lemma says that there exists a code $\cC \subseteq \cU^n$ of rate $I(P_X,W)$ that covers the typical set $\typecl_{P_X,\delta_n}$, in the sense that for every $\bx \in\typecl_{P_X,\delta_n}$, there exists a `codeword' $\bu \in \cC$ such that the sequences $\bx,\bu$ have a joint type that is very close to $(P_X,W)$.

\subsection{The inherently typical subset lemma}
The proof of the converse of Theorem~\ref{thm:RID} relies heavily on the inherently typical subset lemma, due to Ahlswede et al.\cite{Ahlswede97}. An inherently typical set is a generalization of the conditional type class concept, as detailed below. Loosely speaking, the lemma says that every set contains an inherently typical subset of essentially the same size.

Before stating the lemma we require several definitions.
\begin{defn}[Prefix set]
    For $A \subseteq \cX^n$ and $1 \leq k \leq n$, let $A_k$ denote the set of prefixes of sequences of $A$, i.e.
    \begin{equation}
      A_k \triangleq \{x^k \in \cX^k : \exists \tilde x^n \in A \mbox{ s.t. } x^k =\tilde x^k\}.
    \end{equation}
    We denote $A_0 = \{\Lambda\}$, where $\Lambda$ denotes the empty string. Note that, by our convention, for any $\bx \in A$, $x^0 = \Lambda$.
\end{defn}

\smallskip

\begin{defn}[Causal mapping]
    For $A \subseteq \cX^n$ and an arbitrary alphabet $\cU$, a mapping $\phi:A \ra \cU^n$ is said to be \emph{causal}, if there exist mappings $\{\phi_i : A_i \ra \cU\}_{i=0}^{n-1}$, s.t. for all $\bx \in A$ we have
    \begin{equation}
       \phi(\bx) = [\phi_1(x^0),  \phi_1(x^1),...,\phi_n(x^{n-1})].
    \end{equation}
\end{defn}

\smallskip

Let $\cU_m$ denote a discrete alphabet of size
\begin{equation}
    |\cU_m| = |\cP_m(\cX)| = \comb{|\cX| + m - 1}{m}.
\end{equation}
The alphabet $\cU_m$ is exactly the set of $m$-types over the alphabet $\cX$. For example, if $\cX = \{0,1\}$, we have $|\cU_m| = m+1$.

\begin{defn}[Inherently typical set \cite{Ahlswede97}]
    A set $A \subseteq \cX^n$ is said to be $m$-\emph{inherently typical}, if there exist a causal mapping $\phi:A \ra \cU_m^n$ and an $n$-type $Q \in \cP_n(\cX \times \cU_m)$, s.t.
    \begin{enumerate}
      \item For every $\bx \in A$, the sequences $\bx,\phi(\bx)$ have the joint type $Q$, i.e.
          \begin{equation}
            \type_{\bx,\phi(\bx)} = Q.
          \end{equation}
      \item If $X_0,U_0$ are jointly distributed according to $Q$, then
      \begin{equation}
        \frac{1}{n}\log |A| \leq H(X_0 | U_0) \leq \frac{1}{n}\log |A| + \frac{\log^2m}{m}.
      \end{equation}
    \end{enumerate}
\end{defn}

\begin{lem}[The inherently typical subset lemma \cite{Ahlswede97}]
Let $m > 2^{16|\cX|^2}$, and let $n$ be large enough s.t. $(m+1)^{5|\cX|+4}\ln(n+1)/n \leq 1$. Then for any set $A \subseteq \cX^n$ there exists a subset $\tilde A \subseteq A$ that is $m$-inherently
typical, whose size satisfies
  \begin{equation}
    \frac{1}{n}\log \frac{|A|}{|\tilde A|} \leq |\cX|(m+1)^{|\cX|} \frac{\log(n+1)}{n}.
  \end{equation}
\end{lem}
Note that the lemma holds starting at $m=2^{16|\cX|^2}$, which is extremely large, even for $X$ being binary. It is therefore expected that finite blocklength versions of results based on the lemma will be very loose. Nevertheless, it is sufficient for proving a converse for the error exponent (as in \cite{Ahlswede97}), and also for proving the converse for the rate (as we show next).

\subsection{The Blowing Up Lemma}
The blowing up lemma (see, e.g., \cite[Lemma 12]{RaginskySasonConcentration}) will play a key role in proving the converse part of Theorem~\ref{thm:RID}.
\begin{lem}[Blowing-up Lemma]\label{lem:BlowingUp}
    Let $\bZ$ be distributed i.i.d. according to $P_Z$ and let $B \subseteq \cZ^n$ be a given set. Then:
    \begin{align}
      &\Pr\{\bZ \notin \Gamma_H^\delta(B)\} \leq \exp\left[-2n\left(\delta-\sqrt{\frac{1}{2n}\ln\left(\frac{1}{\Pr\{\bZ \in B\}}\right)}\right)^2\right], \\
&\quad\quad\quad \forall \delta >
      \sqrt{\frac{1}{2n}\ln\left(\frac{1}{\Pr\{\bZ \in B\}}\right)}.\nonumber
    \end{align}
\end{lem}

\subsection{Proof of Theorem~\ref{thm:RID}}\label{ssec:proofRIDthm}

As discussed above, the direct part of the theorem follows from the combination of \cite[Thm. 2]{Ahlswede97} and Prop.~\ref{prop:VL=FL} above. Nevertheless, for completeness, we give a direct proof here. It will be simpler than the proof of \cite[Thm. 2]{Ahlswede97}, since it only involves the achievable rate (and not error exponents). After reading through the direct part, it should be easier to the reader to follow the converse proof.

\begin{proof}[Proof of Theorem~\ref{thm:RID}, direct part]

Let $W:\cP(\cX \ra \cU)$ be an arbitrary channel into an arbitrary (discrete) alphabet $\cU$. Following Lemma~\ref{lem:covering}, let $\omega:\typecl_{P_X,\delta_n} \ra \cU^n$ be the mapping that covers the typical set $\typecl_{P_X,\delta_n}$. Let $\cC = \left\{\omega(\bx) : \bx \in \typecl_{P_X,\delta_n} \right\}$ denote the set of codewords. Let $\bu = \omega(\bx)$. The mapping $T(\cdot)$ is defined as follows:

\begin{equation}
  T(\bx) = \left\{
             \begin{array}{ll}
               [\bu,\type_{\bx\bu}], & \hbox{$\bx \in \typecl_{P_X,\delta_n}$;} \\
               \texttt{e} , & \hbox{otherwise.}
             \end{array}
           \right.
\end{equation}
In other words, the signature $T(\bx)$ consists of the (index to the) codeword $\bu$, and also the joint type of $\bx$ and $\bu$. The special symbol $\texttt{e}$ denotes `erasure'.

The decision function $g(T(\bx),\by)$ shall return $\MAYBE$ only if one of the following occurs:
\begin{itemize}
  \item The erasure symbol was received, i.e. $T(\bx)=\texttt{e}$, or
  \item $\by$ was not typical, i.e. $\by \notin \typecl_{P_Y,\delta_n}$, or
  \item $\by$ is typical, $T(\bx) \neq \texttt{e}$, and there exists a joint type $Q\in\cP(\cX\times\cU\times \cY)$ s.t.:
\begin{enumerate} 
  \item The marginal of $Q$ w.r.t. $U,Y$ is $\type_{\bu\by}$ (i.e. the empirical joint distribution of the sequences $\bu,\by$),
  \item the marginal of $Q$ w.r.t. $X,U$ is $\type_{\bx\bu}$ and
  \item the marginal of $Q$ w.r.t. $X,Y$ satisfies  $\EE[\rho(X,Y)]\leq D$.
\end{enumerate}
\end{itemize}

First, let us see why this scheme is admissible, i.e. there are no false negatives. Note that the only case where the scheme may return $\NO$ is when the signature consists of $[\bu,\type_{\bx\bu}]$. If $d(\bx,\by) \leq D$, then the joint type $Q = \type_{\bx\bu\by}$ satisfies the conditions 1) - 3) above, and therefore the scheme would return $\MAYBE$, as required.

Next, note that the rate of the scheme is arbitrarily close to $I(P_X,W)$, since the joint $n$-type $\type_{\bx,\bu}$ adds $O(\log(n)/n)$ to the rate, and the signature $\texttt{e}$ adds a negligible amount to the rate.

The last point that needs to be proved is that the probability of $\MAYBE$ vanishes. Consider the three error events:

\begin{itemize}
  \item $T(\bx) = \texttt{e}$. This happens if and only of $\bx \notin \typecl_{P_X,\delta_n}$, which, following \eqref{eqn:probTypicalSet} and the definition of $\delta_n$, vanishes with $n$.
  \item For a similar reason, the probability that $\by \notin \typecl_{P_Y,\delta_n}$ vanishes with $n$.
  \item The remaining event is when $\bx$ and $\by$ are both typical (i.e. have empirical distributions near $P_X$ and $P_Y$, respectively), and there exists a joint type $Q$  that satisfies 1)-3) above.
\end{itemize}

Fix a sequence $\bx \in \typecl_{P_X,\delta_n}$ and let $P_Y' \in \cP_n(\cY)$ be a given $n$-type. Conditioned on $\bY \in T_{P_Y'}$, we know that $\bY$ is distributed uniformly on $T_{P_Y'}$. Among those sequences, the fraction of these sequences that trigger a $\MAYBE$ are the ones residing in the $Q_{Y|U}(\bu)$-shell of some distribution $Q$ that satisfies 1)-3). Note that the requirement 3) is on the distribution $Q$, and not on the sequence $\by$. Therefore the probability of $\MAYBE$, conditioned on $\bX = \bx$ and $\by \in T_{P_Y'}$, is upper bounded, for any $\eps>0$ and large enough $n$, by
\begin{equation}
  \sum_{
Q:
Q_{XU} = \type_{\bx\bu},
\EE_Q[\rho(X',Y')]\leq D,
Q_{Y} = P_Y'
 } 2^{-n(I(Y';U')-\eps)},
\end{equation}
where $X',U',Y'$ are jointly distributed according to $Q$. Since $P_Y'$ is arbitrarily close to $P_Y$ and $\type_{\bx\bu}$ is arbitrarily close to $(P_X,W)$, we can sum over the (polynomial number of) types $P_Y'$, and bound the probability for $\MAYBE$, now conditioned on $\bx \in T_{P_X,\delta_n}$ and $\by \in T_{P_Y,\delta_n}$, is upper bounded, for any $\eps>0$ and large enough $n$, by
\begin{equation}\label{eqn:thirdErrorEvent}
  \max_{
Q:
Q_{XU} = (P_X,W),
\EE_Q[\rho(X',Y')]\leq D,
Q_{Y} = P_Y
 } 2^{-n(I(Y';U')-\eps)}.
\end{equation}

It therefore remains to show that if $W$ is chosen according to \eqref{eqn:R_ID_info}, then there are no types $Q$ that satisfy 1)-3), for which $U'$ and $Y'$ are independent.  This fact will make the exponent in \eqref{eqn:thirdErrorEvent} to be strictly positive, completing the proof.

In order to show this, let $\Delta_R, \Delta_D >0$ be arbitrarily small constants, s.t.
\begin{equation}
  R = \Delta_R + \bar R_\ID(D-\Delta_D).
\end{equation}
In other words, we know that $(P_X,W)$ satisfies
\begin{equation}\label{eqn:preLongAlign}
  \sum_{u \in \cU} P_U(u) \bar\rho(P_{X|U}(\cdot|u),P_Y) \geq D + \Delta_D.
\end{equation}
and we need to prove that for all joint distributions $Q$ that satisfy $Q_{XU}=(P_X,W)$, $Q_Y=P_Y$, and $\EE_Q[\rho(X,Y)]\leq D$, we must have $I(Y;U) > 0$. To prove this, assume, for contradiction, that $Q$ is such a distribution, where $Y$ and $U$ are independent. Then we can write the following:
\begin{align}
  D & \geq \EE_Q[\rho(X,Y)] \\
  D & \geq \sum_{u \in \cU} P_U(u) \sum_{x,y} Q_{XY|U}(x,y|u) \rho(x,y) \\
  D & \geq \sum_{u \in \cU} P_U(u) \sum_{x,y} Q_{Y|U}(y|u) Q_{X|YU}(x,y|u) \rho(x,y) \\
  D & \geq \sum_{u \in \cU} P_U(u) \sum_{x,y} P_{Y}(y) Q_{X|YU}(x,y|u) \rho(x,y)
\end{align}
Note that for all $u$, the term $P_{Y}(y) Q_{X|YU}(x,y|u)$ defines a joint distribution on $X,Y$, with marginals $P_Y$ and $Q_{X|U}(\cdot|u) = P_{X|U}(\cdot|u)$. Therefore the inner summation is an upper bound, by definition, for the term $\bar\rho( P_{X|U}(\cdot|u),P_Y)$, and we get that
\begin{equation}
  D  \geq \sum_{u \in \cU} P_U(u) \bar\rho( P_{X|U}(\cdot|u),P_Y),
\end{equation}
which contradicts \eqref{eqn:preLongAlign}. Therefore $U$ and $Y$ can never be independent, and the exponent in \eqref{eqn:thirdErrorEvent} is strictly positive. This makes sure that the third error event also has a vanishing probability, and the proof is concluded by using the union bound on the error events.
\end{proof}

\bigskip

To prove the converse part, we start by following the steps similar to that of the converse of \cite[Thm. 2]{Ahlswede97}. In our case these steps are actually slightly simpler than those of \cite[Thm. 2]{Ahlswede97} because of the restriction to fixed-length compression schemes. As mentioned before, the key step that is missing from \cite[Thm. 2]{Ahlswede97} is the ability to transform an event with probability that vanishes with an exponent that is arbitrarily small, to an event whose probability goes to $1$. The key to achieving this is the blowing-up lemma, as detailed in what follows.

\begin{proof}[Proof of Theorem~\ref{thm:RID}, converse part]

Let $\Delta_R,\Delta_D>0$, and let $T,g$ be a sequence of schemes with rate $R < \bar R_\ID(D+\Delta_D) - \Delta_R$. Our goal is to show that the probability for $\MAYBE$ cannot vanish with $n$, for arbitrarily small $\Delta_R$ and $\Delta_D$.

Let $i \in [1:2^{nR}]$, and let
\begin{equation}
  T^{-1}(i) \triangleq \{\bx \in \cX^n : T(\bx) = i\}.
\end{equation}

Since the given scheme is admissible, we must have
\begin{align}
    \Pr\{\MAYBE\}
    &= \sum_{i=1}^{2^{nR}} \Pr\{T(\bX) = i\} \Pr\{\MAYBE | T(\bX) = i\}\\
    & \geq \sum_{i=1}^{2^{nR}} \Pr\left\{T(\bX) = i\right\} \Pr\left\{\bY \in \Gamma^D\left(T^{-1}(i)\right)\right\}.
\end{align}

For some $\gamma>0$, define the set
\begin{equation}
  A_i \triangleq T^{-1}(i) \cap \typecl_{P_X,\gamma},
\end{equation}
Since by definition $A_i \subseteq T^{-1}(i)$, we also have
\begin{align}
    \Pr\{\MAYBE\}
    & \geq \sum_{i=1}^{2^{nR}} \Pr\left\{T(\bX) = i\right\} \Pr\left\{\bY \in \Gamma^D(A_i) \right\}.
\end{align}

It appears that it is enough to consider only $A_i$'s that are `large', by the following lemma:
\begin{lem}[Most $A_i$'s are large]\label{lem:MostAiAreLarge}
    There exists $n_0 = n_0(\gamma) >0$, s.t. for all $n>n_0$,
    \begin{equation}
      \sum_{i: |A_i| \leq 2^{n[H(P_X)-R -2\gamma']} }\Pr\{\bX \in A_i\} \leq 2^{-\gamma' n },
    \end{equation}
    where $\gamma' \triangleq 2 \gamma |\cX| \log (1/\gamma)$.
\end{lem}

\begin{proof} Appendix~\ref{app:MostAiAreLarge}.\end{proof}
Next, consider a specific $A = A_i$, and suppose that $|A| \geq 2^{n(H(X)-R-2\gamma')}$ (by the previous lemma we know that this occurs with high probability). We invoke the inherently typical subset lemma and conclude that for any $m > 2^{16 |\cX|^2}$ and large enough $n$, there exists a subset $\tilde A \subseteq A$, for which:
\begin{enumerate}
  \item The size of the set $\tilde A$ is essentially the same as $A$:
  \begin{equation}\label{eqn:ITSproprtey1}
    \frac{1}{n}\log \frac{|A|}{|\tilde A|} \leq |\cX|(m+1)^{|\cX|} \frac{\log(n+1)}{n}.
  \end{equation}
  \item There exists an $n$-type $Q \in \cP_n(\cX \times \cU_m)$, and a causal mapping $\phi:\cX^n \ra \cU_m^n$, s.t. for every $\bx \in \tilde A$,
      \begin{equation}\label{eqn:ITSproprtey2}
        \type_{\bx,\phi(\bx)} = Q.
      \end{equation}
  \item If $X_0,U_0$ are jointly distributed according to $Q$, then
  \begin{equation}\label{eqn:ITSproprtey3}
    \frac{1}{n}\log |\tilde A| \leq H(X_0 | U_0) \leq \frac{1}{n}\log |\tilde A| + \frac{\log^2m}{m}.
  \end{equation}
\end{enumerate}

Since $A \subseteq T_{P_X,\gamma}$, we must have that the marginal of $Q$ w.r.t. $X$ has a type $P$ that satisfies $\|P-P_X\|_\infty \leq \gamma$.

Since $\tilde A \subseteq A$, we have
\begin{equation}
  \Pr\{\bY \in \Gamma^D A \} \geq  \Pr\{\bY \in \Gamma^D \tilde A \}.
\end{equation}

\bigskip

Let $\eps >0$ and let $V:\cX\times \cU_m \ra \cY$ be a stochastic matrix, s.t.
\begin{equation}\label{eqn:propertyV}
  \sum_{x,y,u} Q(x,u) V(y|x,u) \rho(x,y) \leq D - 2\eps.
\end{equation}
In other words, if $X_0,U_0,Y_0$ are distributed according to $(Q,V)$, then $\EE[\rho(X_0,Y_0)]\leq D-2\eps$.


Following \eqref{eqn:propertyV}, for all $\bx \in \tilde A$ and large enough $n$, $\by \in T_{Y_0|X_0U_0,\delta_n}(\bx,\phi(\bx))$ implies $d(\bx,\by)\leq D-\eps$.
Next, define the set $F$ to be the union of all such conditional type classes as follows:
\begin{equation}
  F \triangleq \bigcup_{\bx \in \tilde A} T_{Y_0|X_0U_0,\delta_n}(\bx,\phi(\bx)).
\end{equation}
Clearly, from the above it follows that $F \subseteq \Gamma^{D-\eps}\left(\tilde A\right)$.

Let $\eps' \triangleq \eps/\rho_{\max}$. Note the following fact:
\begin{equation}\label{eqn:MixedDistortions}
  \Gamma_H^{\eps'} \left(\Gamma^{D-\eps}\left(\tilde A\right)\right) \subseteq \Gamma^{D}\left(\tilde A\right).
\end{equation}
To see this, suppose $\by \in \Gamma_H^{\eps'} \left(\Gamma^{D-\eps}\left(\tilde A\right)\right)$.
Then there exists $\by' \in \Gamma^{D - \eps} (\tilde A)$ s.t. $d_H(\by,\by') \leq \eps'$. Also, there exists $\bx \in \tilde A$ s.t. $d(\bx,\by') \leq D - \eps$. Now write
    \begin{align}
      n\cdot d(\bx, \by)
      & = \sum_{i=1}^n d(x_i,y_i) \\
      & = \sum_{i : y_i = y_i'} d(x_i,y_i') + \sum_{i : y_i \neq y_i'} d(x_i,y_i') \\
      & \leq \sum_{|i=1}^n d(x_i,y_i') + \sum_{i : y_i \neq y_i'} \rho_{\max} \\
      & \leq n(D-\eps) + n \eps' \rho_{\max} \\
      & =  nD.
    \end{align}
With \eqref{eqn:MixedDistortions} above we have
\begin{align}
  \Pr\left\{\by \in \Gamma^D(A)\right\} & \geq \Pr\left\{\by \in \Gamma_H^{\eps'}\left(F\right)\right\}.
\end{align}
Next, we apply the blowing up lemma to the set $F$, and have:
\begin{align}
\Pr\left\{\by \in \Gamma_H^{\eps'}\left(F\right)\right\}
  & \geq 1 - \exp\left[-2n\left(\eps'-\sqrt{\frac{1}{2n}\ln\left(\frac{1}{\Pr\{\bY \in F\}}\right)}\right)^2\right],
\end{align}
whenever
\begin{equation}\label{eqn:blowingUpCondition}
  \eps'>\sqrt{\frac{1}{2n}\ln\left(\frac{1}{\Pr\{\bY \in F\}}\right)}.
\end{equation}
Since $\eps'>0$ is a constant, all that is left is to prove that $\Pr\{\bY \in F\}$ does not vanish exponentially. Note that we only need to show the existence of a single channel $V$ for which $\Pr\{\bY \in F\}$ does not vanish exponentially. To show this, we follow the steps of \cite{Ahlswede97}:

First, since $F$ is a union of $V$-shells, we deduce that all the sequences in it are typical:
\begin{equation}
  F \subseteq T_{Y_0,\delta_n\cdot |\cX||\cU_m|},
\end{equation}
where $Y_0$ is an RV with the marginal distribution of $(Q,V)$. It follows that
\begin{equation}\label{eqn:PrYinF}
  \Pr\{\bY \in F\} \geq \frac{|F|}{2^{n H(Y_0)}} 2^{-n(D(Y_0||Y) + \xi_n)},
\end{equation}
where $\xi_n \ra 0$ with $n$. We see that the key to evaluating $\Pr\{\bY \in F\}$ is evaluating $|F|$.

\bigskip

Let $\tilde \bX$ be uniformly distributed on $\tilde A$, and let $\tilde \bU \triangleq \phi(\tilde\bX)$. Note that $\tilde \bX,\tilde \bU$ are random vectors which are not i.i.d., but have the joint type $Q$ with probability $1$. Next, define $\tilde \bY$ to be the output of the memoryless channel $V$ with inputs $\tilde\bX,\tilde\bU$. Here are several facts regarding the random vector $\tilde \bY$, expressed by the random variables $X_0,U_0,Y_0$ that are distributed according to $(Q,V)$.
\begin{itemize}
  \item Since $\tilde \bY$ is conditionally typical w.h.p. (given $\tilde \bX = \bx \in \tilde A$), \begin{align}
    \Pr\{\tilde\bY \in F | \tilde \bX = \bx\}
    & \geq \Pr\{\tilde \bY \in T_{Y_0|X_0U_0,\delta_n}(\bx,\phi(\bx)) | \tilde \bX = \bx\} \\
    & \geq 1 - \tfrac{|\cX||\cY||\cU_m|}{4n\delta_n}.
  \end{align}
  \item Since the above holds for any $\bx \in \tilde A$, we also have
      \begin{align}
    \Pr\{\tilde\bY \in F \}
    &= \sum_{\bx \in \tilde A} \Pr\{\tilde\bX = \bx\} \Pr\{\tilde\bY \in F | \tilde \bX = \bx\} \nonumber\\
%
    & \geq 1 - \tfrac{|\cX||\cY||\cU_m|}{4n\delta_n}.\label{eqn:PrtildeYinF}
  \end{align}
  %
\end{itemize}

\bigskip

For convenience, let $\gamma_n \triangleq \frac{|\cX| |\cY|}{4n\delta_n^2}$. Note that by the delta convention, we have $\gamma_n \ra 0$. Define the indicator RV $\chi_F$ as
\begin{equation}\
  \chi_F(\tilde \bY) \triangleq 1\{\tilde \bY \in F\}
\end{equation}

It follows that
\begin{align}
  H(\tilde \bY)
  & = H\left(\tilde \bY , \chi_F(\tilde \bY)\right) \\
  & = H\left(\tilde \bY | \chi_F(\tilde \bY)\right) + H(\chi_F(\tilde\bY))\\
  & = H\left(\tilde \bY | \chi_F(\tilde \bY)\right) + h(\Pr\{\tilde\bY \in F\})\\
  & \overset{(a)} \leq H\left(\tilde \bY | \chi_F(\tilde \bY)\right) + h(\gamma_n)\\
  & = H\left(\tilde \bY | \tilde\bY \in F \right)\Pr\{\tilde\bY \in F\} +
  H\left(\tilde \bY | \tilde\bY \notin F \right)\Pr\{\tilde\bY \notin F\}
   + h(\gamma_n)\\
  & \leq H\left(\tilde \bY | \tilde\bY \in F \right) +
  n \log |\cY| \Pr\{\tilde\bY \notin F\}
   + h(\gamma_n)\\
  & \leq \log |F| +
  n \log |\cY| \Pr\{\tilde\bY \notin F\}
   + h(\gamma_n)\\
  & \leq \log |F| +
  \gamma_n n \log |\cY|
   + h(\gamma_n). \label{eqn:HofYtildeBound}
\end{align}

where $(a)$ follows from \eqref{eqn:PrtildeYinF} for $n$ large enough s.t. $\gamma_n < 1/2$.

The last derivation reveals that we can bound $H(\tilde \bY)$ to get a lower bound on $|F|$ (and by that a lower bound on $\Pr \{\bY \in \Gamma^D(A)\}$).

The next step follows \cite[(4.26)-(4.30)]{Ahlswede97} almost verbatim. For completeness, we pack the argument into a lemma, and prove it in the appendix (with simpler notation than in \cite{Ahlswede97}).
\begin{lem}\label{lem:LogSumLemma}
    With $\tilde\bY$ defined above, we have
    \begin{equation}
      \frac{1}{n}H(\tilde\bY) \geq H(Y_0 | U_0) - \frac{\log^2m}{m}.
    \end{equation}
\end{lem}
\begin{proof} Appendix~\ref{app:LogSumLemma}.\end{proof}
We conclude that
\begin{equation}
   \frac{1}{n}\log|F| \geq H(Y_0|U_0) - \frac{\log^2m}{m} - \tfrac{1}{n}h(\gamma_n) - \gamma_n \log|\cY|.
\end{equation}
Substituting back into \eqref{eqn:PrYinF} gives
\begin{align}
  \Pr\{\bY \in F\}
  &\geq \frac{|F|}{2^{n H(Y_0)}} 2^{-n(D(Y_0||Y) + \xi_n)}\\
  &\geq 2^{-n \left[I(Y_0;U_0)+D(Y_0||Y) + \xi_n + \tfrac{\log^2m}{m}
+\tfrac{1}{n}h(\gamma_n) + \gamma_n \log|\cY|
  \right]}.\label{eqn:ExponentThatNeedsToBeSmall}
\end{align}
Our next step is to show that the exponent in \eqref{eqn:ExponentThatNeedsToBeSmall} can be made arbitrarily small, by a proper selection of the channel $V$. For that purpose, write:
\begin{align}
  I(X_0;U_0)
  & \leq H(X_0)-H(X_0|U_0) \\
  & \overset{(a)}\leq H(X)-H(X_0|U_0) + \gamma'\\
  & \overset{(b)}\leq H(X)-\frac{1}{n}\log|\tilde A| + \gamma'\\
  & \overset{(c)}\leq H(X)-\frac{1}{n}\log|A| + \gamma' + \zeta_n\\
  & \overset{(d)}\leq R+3\gamma' + \zeta_n\\
  & \overset{(e)}\leq \bar R_\ID(D-\Delta_D) - \Delta_R +3\gamma' + \zeta_n.
\end{align}
In the above, $(a)$ follows by \cite[Lemma 2.7]{CsiszarKorner2nd} (cf. also the proof of Lemma~\ref{lem:MostAiAreLarge}). $(b)$ and $(c)$ follow from \eqref{eqn:ITSproprtey3} and  \eqref{eqn:ITSproprtey1} respectively, where  $\zeta_n \triangleq |\cX|(m+1)^{|\cX|} \frac{\log(n+1)}{n}$. $(d)$ follows from the assumption that $|A| \geq 2^{n(H(X)-R-2\gamma')}$, and $(e)$ follows from the assumption at the beginning of the proof.

Next, we use the fact that $X$ and $X_0$ have distributions that are very close (with closeness quantified by $\gamma$), and write:
\begin{align}
  \bar R_\ID(D-\Delta_D)
  & = \min_{P_{U|X}: \sum_{u\in\cU} P_U(u) \bar\rho(P_{X|U}(\cdot|u),P_Y)\geq D - \Delta_D } I(X;U),\\  & \leq \min_{P_{U|X_0}: \sum_{u\in\cU} P_U(u) \bar\rho(P_{X_0|U}(\cdot|u),P_Y)\geq D - \Delta_D + \gamma''} I(X_0;U) + \gamma'',
\end{align}
for some $\gamma''>0$ that vanishes with $\gamma$.

Next, let $\eps$, defined above to be arbitrarily small but positive, take the value of $\frac{1}{3}\Delta_D$. Also let $\gamma$ be small enough s.t. $\gamma'' < \eps$, and that $3\gamma' + \gamma'' \leq \frac{1}{2}\Delta_R$. This way, whenever $n$ is large enough so that $\zeta_n \leq \frac{1}{2}\Delta_R$, we have
\begin{equation}
  I(X_0;U_0) \leq \min_{P_{U|X_0}: \sum_{u\in\cU} P_U(u) \bar\rho(P_{X_0|U}(\cdot|u),P_Y)\geq D  -2\eps} I(X_0;U).
\end{equation}
From the equation above, we deduce that
\begin{equation}
  \sum_{u\in\cU_m} P_{U_0}(u) \bar\rho(P_{X_0|U_0}(\cdot|u),P_Y) < D  -2\eps.
\end{equation}
By the definition of $\bar\rho(\cdot,\cdot)$, there exist distributions $\Psi_u(x,y)$, for each $u\in\cU_m$, s.t.
\begin{equation}
  \sum_{u\in\cU_m} P_{U_0}(u) \EE_{\Psi_u}[\rho(X,Y)] \leq D-2\eps.
\end{equation}
Furthermore, the marginals of $\Psi_u$ are $P_{X_0|U_0}(\cdot|u)$ and $P_Y$. With $P_{U_0}$, $\Psi_u$ defines a joint distribution with the following properties:
\begin{itemize}
  \item The marginal of the distribution w.r.t. $X_0,U_0$ is exactly $Q$. Therefore the conditional distribution w.r.t. $Y$ is a feasible choice for $V$ according to \eqref{eqn:propertyV}. Denote the RV at the output of this channel by $Y_0$.
  \item The marginal distribution w.r.t. $U_0,Y_0$ is given by $P_{U_0} \times P_Y$, i.e. $U_0$ and $Y_0$ are \emph{independent}, and we also have that $P_{Y_0} = P_Y$.
\end{itemize}
Indeed, we choose $V$ to be defined according to $P_{U_0},\Psi_u$, and conclude that
\begin{equation}
  I(Y_0;U_0)+D(Y_0||Y) = 0.
\end{equation}
Substituting back into \eqref{eqn:ExponentThatNeedsToBeSmall}, we see that the desired exponent is arbitrarily close to $\frac{\log^2m}{m}$. We then set $m$ to be large enough s.t. the condition \eqref{eqn:blowingUpCondition} holds.

To summarize, so far we have shown that if $|A| \geq 2^{n(H(X)-R-2\gamma')}$, then
\begin{equation}\label{eqn:upperBoundForLargeA}
  \Pr\{\bY \in \Gamma^D(A)\} \geq \eta_n,
\end{equation}
where $\eta_n$ approaches $1$ (exponentially fast), as a result of the blowing up lemma.

Finally, we repeat this for each of the sets $A_i$, and write:
\begin{align}
    \Pr\{\MAYBE\}
    & \geq \sum_{i=1}^{2^{nR}} \Pr\left\{T(\bX) = i\right\} \Pr\left\{\bY \in \Gamma^D(A_i) \right\}\\
    & \geq \sum_{i : |A_i| \geq 2^{n(H(X)-R-2\gamma')}} \Pr\left\{T(\bX) = i\right\} \Pr\left\{\bY \in \Gamma^D(A_i) \right\}\\
    & \overset{(a)}\geq \sum_{i : |A_i| \geq 2^{n(H(X)-R-2\gamma')}} \Pr\left\{T(\bX) = i\right\} \eta_n\\
    & = \eta_n \sum_{i : |A_i| \geq 2^{n(H(X)-R-2\gamma')}} \Pr\left\{T(\bX) = i\right\} \\
    & \overset{(b)}\geq \eta_n \left[ 1- \sum_{i : |A_i| \leq 2^{n(H(X)-R-2\gamma')}}\Pr\left\{T(\bX) = i\right\} \right]\\
    & \geq \eta_n \left[ 1- 2^{-n\gamma'} \right].
\end{align}
In the above, $(a)$ follows from \eqref{eqn:upperBoundForLargeA} and $(b)$ follows from Lemma~\ref{lem:MostAiAreLarge}. Since $\eta_n$ approaches $1$ exponentially fast, we conclude that the probability for $\MAYBE$ approaches $1$, also exponentially fast. This concludes the proof of the converse.
\end{proof}

\section{Schemes based on the triangle inequality}\label{sec:Triangle}

In this section we discuss the triangle-inequality based schemes: the lossy compression - triangle inequality (\LCT) and the type covering - triangle inequality (\TCT).

\subsection{Lossy compression with triangle inequality (\LCT)}
Here we prove that any compression rate above $R^\LCT_\ID(D)$ [defined as the inverse function of $D^\LCT_\ID(R)$, see \eqref{eqn:DIDLCT}] can be attained via a scheme that employs standard lossy compression for the signature assignment and the triangle inequality for the decision rule.

\begin{proof}[Proof of Theorem~\ref{thm:Naive}]
We will show that any pair $(R,D)$ s.t. $D > D^\LCT_\ID(R)$ is achievable, where
\begin{equation}\label{eqn:NaiveRepeated}
  D^\LCT_\ID(R) = \EE[\rho(\hat X,Y)]-D(R),
\end{equation}
where $D(R)$ is the classical distortion-rate function of the source $X$, and $\hat X$, which is independent of $Y$, is distributed according to any marginal distribution of the $D(R)$ achieving distribution.

Let $P_{X|\hat X}$ be an achieving distribution for the standard distortion rate function at rate $R$, and let $\hat X$ be the corresponding marginal distribution. Next, use the covering lemma (Lemma~\ref{lem:covering}) to show the existence of a code $\cC$ that covers the typical set%
\footnote{Recall that $\delta_n$ is defined according to the delta-convention \cite{CsiszarKorner2nd}. See also Sec.~\ref{sec:ProofConverse}.} $\typecl_{P_X,\delta_n}$ with the distribution $P_{X\hat X}$. In other words, for each sequence in the typical set $\bx\in \typecl_{P_X,\delta_n}$, there exists a sequence $\omega(\bx) = \hat \bx \in \cC$ s.t. $\bx,\hat\bx$ are strongly jointly typical according to the distribution $P_{X\hat X}$ (formally,
$ \hat\bx \in \typecl_{V,\delta_n}(\bx)$). We also know by the covering lemma that the code rate is upper bounded by $I(X;\hat X) + \eps_n$ where $\eps_n$ vanishes as $\delta_n \ra 0$. Since $P_{X,\hat X}$ is the distortion-rate achieving distribution, we know that $d(\bx,\hat\bx) \leq D(R) + \eps'_n$, with some $\eps'_n$ that vanishes as $\delta_n \ra 0$.

So far, we have constructed a standard code for lossy compression: for an input $\bx$, if it is typical, then its compressed representation is the index to $\hat\bx$. If $\bx$ is not typical, then declare an `erasure' $\texttt{e}$. Therefore with probability approaching one, we have a guarantee that the distortion between the source and the reconstruction is at most $D(R)$.

Next, we use this code in order to construct a compression scheme for identification. We only need to specify the decision process $g(\cdot,\cdot)$, which proceeds as follows. If $T(\bx) = \texttt{e}$, we set $g(T(\bx),\by)= \MAYBE$. Otherwise, we reconstruct $\hat\bx$, and have
\begin{equation}
  g(T(\bx), \by) = \left\{
                     \begin{array}{ll}
                       \MAYBE, & \hbox{if $d(\hat\bx,\by) \leq D + D(R) + \eps'_n$ ;} \\
                       \NO, & \hbox{otherwise.}
                     \end{array}
                   \right.
\end{equation}
It follows from the triangle inequality that whenever $d(\bx,\by) \leq D$ and when $d(\bx,\hat\bx) \leq D(R) + \eps'_n$, then $d(\hat\bx,\by) \leq D + D(R) + \eps'_n$, triggering a $\MAYBE$. Therefore the scheme is admissible.

Since, by construction, the rate of the scheme is arbitrarily close to $R$, we only need to verify that the probability of $\MAYBE$ vanishes.

Next, assume that the similarity threshold $D$ satisfies $D = D^\LCT_\ID(R)=\Delta_D$ for some $\Delta_D > 0$. We analyze the probability of $\MAYBE$ as follows:
\begin{equation}\label{eqn:asdasdasd}
  \Pr\{\MAYBE\} \leq \Pr\{\bX \notin \typecl_{P_X,\delta_n}\} + \Pr\{\bY \notin \typecl_{P_Y,\delta_n}\} + \Pr\{\MAYBE | \bX \in \typecl_{P_X,\delta_n}, \bY \in \typecl_{P_Y,\delta_n}\}.
\end{equation}
The first two terms in the summation vanish with $n$. To bound the third term, we need to evaluate the probability that a sequence $\bY$ will be in a `ball' of radius $D + D(R) + \eps'_n$ centered at $\hat\bx$. Let $\hat\bx$ be the reconstruction sequence with type $P'_{\hat X}$. We know that it has a type close to $P_{\hat X}$, the marginal of the $D(R)$-achieving distribution. Also, let $P_Y' \in \cP_n(\cY)$ be a given $n$-type. Conditioned on $\bY \in T_{P_Y'}$, $\bY$ is now distributed uniformly on $T_{P_Y'}$. Among those sequences, the fraction of sequences that trigger a $\MAYBE$ is given (up to sup-exponential factors) by
\begin{equation}\label{eqn:asdasd}
  \sum \frac{2^{n H(Y|\hat X)}}{2^{n H(Y)}},
\end{equation}
where the summation is over all joint $n$-types for $Y,\hat{X}$ with marginals $P_Y',P'_{\hat X}$ s.t.
$$\EE[\rho(\hat X,Y)] \leq D + D(R) + \eps_n.$$
The exponent of this expression is given by
\begin{equation}
  \min I(Y;\hat X),
\end{equation}
where the minimization is the same as in \eqref{eqn:asdasd}. We can see that this exponent can be made strictly positive: if it was zero, this would imply that there exist independent $\hat X,Y$ s.t.
$$\EE[\rho(\hat X,Y)] > D + D(R) + \eps_n,$$
which contradicts the assumption that $D > D_\ID^\LCT(R)$. The next step follows by standard type arguments showing that $P_Y'$ and $P_{\hat X}'$ are arbitrarily close to the real $P_Y,P_{\hat{X}}$ (see also the direct part of the proof of Theorem~\ref{thm:RID}). Finally, we see that all three terms in the expression for the probability of $\MAYBE$ in \eqref{eqn:asdasdasd} vanish with $n$, as required.
\end{proof}

\subsection{Type covering with triangle inequality (\TCT)}
Here we prove that for any similarity threshold $D > D_\ID^\TCT(R)$, there exists a sequence of rate-$R$ schemes that are $D$-admissible.

\begin{proof}[Proof of Theorem~\ref{thm:LessNaive}]
    The proof follows the steps of the proof of Theorem~\ref{thm:Naive} above, with the following key exception. The conditional distribution determining the code that describes $X$ remains a design parameter and is optimized at the end (rather than at the beginning of the proof as in Theorem~\ref{thm:Naive}, where it is set to be the one that minimizes the expected distortion between $X$ and $\hat X$).

More formally, let $P_{\hat X|X}$ be a conditional distribution s.t. $I(X;\hat X) > R$.
Again, we use the covering lemma (Lemma~\ref{lem:covering}) to show the existence of a code $\cC$ that covers the typical set $\typecl_{P_X,\delta_n}$ with the distribution $P_{X\hat X}$. In other words, for each $\bx\in \typecl_{P_X,\delta_n}$, there exists a sequence $\omega(\bx) = \hat \bx \in \cC$ s.t. $\bx,\hat\bx$ are strongly jointly typical according to the distribution $P_{X\hat X}$ (formally, $ \hat\bx \in \typecl_{V,\delta_n}(\bx)$). Again, we know by the covering lemma that the code rate is upper bounded by $I(X;\hat X) + \eps_n$ where $\eps_n$ vanishes as $\delta_n \ra 0$. Since $\bx,\hat\bx$ are strongly jointly typical, we know that $d(\bx,\hat\bx) \leq \EE[\rho(X,\hat X)] + \eps'_n$, with some $\eps'_n$ that vanishes as $\delta_n \ra 0$.

The rest of the proof is identical to that of Theorem~\ref{thm:Naive}, where $D(R)$ is replaced by $\EE[\rho(X,\hat{X})]$. If we choose $P_{\hat X|X}$ s.t. $\EE[\rho(\hat X,Y)]-\EE[\rho(X,\hat X)] \geq D$, we can verify that the exponent of the third term in the equivalent of \eqref{eqn:asdasdasd} is positive, proving that the overall probability of $\MAYBE$ vanishes, as required.
\end{proof}

Remarks:
\begin{itemize}
  \item It is obvious that $R_\ID^\TCT(D) \leq R_\ID^\LCT(D)$, since the distortion-rate achieving distribution in \eqref{eqn:DIDLCT} is a feasible transition probability for the expression in \eqref{eqn:DIDTCT}. Therefore the \TCT\ scheme can be regarded as a generalization of \LCT.
  \item In order to simplify the discussion, we have assumed that the distortion measure $\rho(\cdot,\cdot)$ is \emph{symmetric}. Similar results can be obtained for the non-symmetric case, where the only difference is that the triangle inequality is applied in the following form:
\begin{equation}
  d(\hat \bx,\by) \leq d(\hat\bx,\bx) + d(\bx,\by).
\end{equation}
Therefore the compression of $\bx$ needs to be done (in the \LCT\ scheme) for the distortion measure $\rho'(x,\hat x) \triangleq \rho(\hat x,x)$. The decision rule (in the typical case) is given by
\begin{equation}
  g(T(\bx), \by) = \left\{
                     \begin{array}{ll}
                       \MAYBE, & \hbox{if $d(\hat\bx,\by) \leq D + \EE[\rho(\hat X,X)] + \eps'_n$ ;} \\
                       \NO, & \hbox{otherwise.}
                     \end{array}
                   \right.
\end{equation}
    \item If, in addition to the symmetry condition, we would require that $\rho(x,y)=0$ if and only if $x=y$, this would make the measure a \emph{metric}. However, there is no need for this third condition in order for the results to hold.
    \item In principle, one could add another condition that rules out sequences that are not similar, using the modified decision rule:
\begin{equation}
  g(T(\bx), \by) = \left\{
                     \begin{array}{ll}
                       \MAYBE, & \hbox{if $D(R)-D \leq d(\hat\bx,\by) \leq D + D(R)$ ;} \\
                       \NO, & \hbox{otherwise.}
                     \end{array}
                   \right.
\end{equation}
This condition retains the admissibility of the scheme by another usage of the triangle inequality. In essence, it allows us to rule out sequences $\by$ that are too close to $\hat \bx$, since we know that $\bx$ is at distance approximately $D(R)$ from it. A similar argument holds for the \LCT\ scheme as well. However, this condition does not improve the achievable rate, and even in practice, the performance gain is generally negligible (see \cite{IdoiaAllerton2013}).
\end{itemize}

\subsection{Special cases}
In general,
\begin{equation}\label{eqn:generalIneq}
      R_\ID(D) \leq R_\ID^\TCT(D) \leq R_\ID^\LCT(D),
\end{equation}
where the inequalities may be strict (see Fig.~\ref{fig:ternary}). There are cases, however, where some of the inequalities in \eqref{eqn:generalIneq} are \emph{equalities}. We review some of those cases here.

\begin{theorem}\label{thm:D0}
    If there exists a constant $D_0$ (that may depend on $P_Y$), s.t. for all $\hat x \in \cX$
    \begin{equation}
        \sum_{y\in\cX} P_Y(y) \rho(\hat x,y) = D_0,
    \end{equation}
    then
\begin{equation}
  R_\ID^\TCT(D) = R_\ID^\LCT(D) = R(D_0-D),
\end{equation}
where $R(\cdot)$ is the rate-distortion function under distortion measure $\rho(\cdot,\cdot)$.
\end{theorem}

\begin{proof}
    Under the stipulation in the theorem, we have that for any RV $\hat X$ that is independent of $Y$,
    \begin{align}
      \EE[\rho(\hat X,Y)]
        & =\sum_{\hat x,y} P_Y(y)P_{\hat X}(\hat x) \rho(\hat x,y) \\
        & =\sum_{\hat x} P_{\hat X}(\hat x) \sum_{y} P_Y(y) \rho(\hat x,y) \\
       & = D_0.
    \end{align}
    It follows that
\begin{align}
    D^\TCT_\ID(R)
    &= \max_{P_{\hat X|X} : I(X;\hat X) \leq R} D_0-\EE[\rho(X,\hat X)]\\
    &= D_0 - \min_{P_{\hat X|X} : I(X;\hat X) \leq R} \EE[\rho(X,\hat X)]\label{eqn:RLCTwannabe}\\
    &= D_0 - D(R).
\end{align}
    The proof follows by noting that \eqref{eqn:RLCTwannabe} is equal to $D_\ID^\LCT(R)$.
\end{proof}

The conditions for the above theorem holds, for example, in the following cases:
\begin{itemize}
  \item If $Y$ is equiprobable on $\cY=\cX$, and the columns of $\rho(\cdot,\cdot)$ are permutations of each other (e.g. if $\rho(\cdot,\cdot)$ is a `difference' distortion measure), then $D_0$ is given by
    \begin{equation}
        D_0 = \sum_{y} \frac{1}{|\cX|} \rho(y,1).
    \end{equation}
  \item A special case of the above is the Hamming distortion. In this case, (where $P_Y$ is still equiprobable),
      \begin{equation}
        D_0 = \frac{|\cX|-1}{|\cX|}.
      \end{equation}
\end{itemize}

\medskip

Theorem~\ref{thm:D0} implies that in simple cases the \LCT\ scheme is equivalent (in the rate sense) to the \TCT\ scheme. If $X$ and $Y$ are binary and equiprobable, and the distortion measure is Hamming, it follows from \cite[Theorem 1]{IdoiaAllerton2013} that $R_\ID^\TCT(D) = R_\ID(D)$, i.e. the \TCT\ and the \LCT\ schemes are optimal. However, if $X$ and $Y$ are not equiprobable (and the distortion measure is still Hamming), the \LCT\ scheme differs from the \TCT\ scheme (see \cite[Fig. 2]{idoiaDCC2014}). Is this \TCT\ scheme optimal for the binary-Hamming case? The following theorem answers this question in the affirmative.

\begin{theorem}\label{thm:RIDbinHamming}
    For the binary-Hamming case, i.e. $X \sim \Ber(p)$ and $Y\sim\Ber(q)$,
    \begin{equation}
      R_\ID(D) = R_\ID^\TCT(D).
    \end{equation}
\end{theorem}
\begin{proof}
    We first show that it is sufficient to take the cardinality of $\cU$ in \eqref{eqn:R_ID_info} to be equal to $2$. To that end, note that for the binary-Hamming case we have $\bar\rho(p,q) = |p-q|$.

    Let $\cU$ be an arbitrary (but finite) alphabet, and let $P_{U|X}$ be a given channel from $\cX$ to $\cU$, that attains the minimum in $R_\ID(D)$. It has to satisfy
    \begin{equation}
      \sum_{u\in\cU} P_U(u) \bar\rho(\Wrev(\cdot|u),P_Y)\geq D,
    \end{equation}
    i.e. it is feasible for the optimization in $R_\ID(D)$. Suppose there exist $u_1$ and $u_2$ for which both $\Wrev(1|u_1) \geq q$ and $\Wrev(1|u_2) \geq q$ hold. Next, define a new channel $P'_{U|X}$ that is the result of the channel $P_{U|X}$, followed by a merge of $u_1$ and $u_2$ into a new symbol $u^*$. By the data processing inequality, the mutual information does not increase (in fact, it decreases unless $\Wrev(1|u_1) = \Wrev(1|u_2)$). The new reverse channel given $u^*$ is easily calculated as
    \begin{equation}
      \Wrev'(1|u^*) = \frac{P_U(u_1)\Wrev(1|u_1)+P_U(u_2)\Wrev(1|u_2)}{P_U(u_1)+P_U(u_2)},
    \end{equation}
    and the new prior is
    \begin{equation}
      P_U'(u^*) = P_U(u_1) + P_U(u_2).
    \end{equation}

    Next, observe that
    \begin{align*}
      &P_U(u_1) \bar\rho(\Wrev(\cdot|u_1),P_Y)+P_U(u_2) \bar\rho(\Wrev(\cdot|u_2),P_Y) \\
      &= P_U(u_1) \left|\Wrev(\cdot|u_1)-q\right| +P_U(u_2) \left|\Wrev(\cdot|u_2)-q\right| \\
      & = P_U(u_1) \left(\Wrev(\cdot|u_1)-q\right) +P_U(u_2) \left(\Wrev(\cdot|u_2)-q\right) \\
      & = P_U'(u^*) \left(\Wrev'(\cdot|u^*)-q\right).
    \end{align*}
    Since for all the other values of $u \in \cU$, $P_{X|U}(x|u) = P'_{X|U}(x|u)$, we conclude that the new channel is also feasible, and attains lower mutual information. Therefore the cardinality of $\cU$ can be safely reduced by $1$, assuring that the value of the optimization of $R_\ID(D)$ will not be higher because of this reduction. The same holds for the case where both $\Wrev(1|u_1) \leq q$ and $\Wrev(1|u_2) \leq q$ hold. This process can be applied for any channel $W$ with more than one value of $u$ for which $\Wrev(1|u)$ at the same side of $q$. Therefore for the optimal channel, it suffices to check only channels with $|\cU|=2$.

\medskip

Next, we aim to show that
\begin{equation}\label{eqn:sameObjDiffFeas}
  \min_{P_{U|X}: \sum_{u\in\cU} P_U(u) \bar\rho(\Wrev(\cdot|u),P_Y)\geq D} I(X;U) = \min_{P_{U|X}: \EE[\rho(U,Y)]-\EE[\rho(X,U)] \geq D } I(X;U).
\end{equation}
Since it suffices to look at binary $U$, the feasibility condition on the LHS of \eqref{eqn:sameObjDiffFeas} can be rewritten as
    \begin{equation}\label{eqn:R2feasible}
        P_U(0) \left|\Wrev(1|0)-q\right|+
        P_U(1) \left|\Wrev(1|1)-q\right|
        \geq D.
    \end{equation}

    \medskip

    Let $P^*_{X|U}$ be a minimizing channel for the LSH of \eqref{eqn:sameObjDiffFeas}. Denote the reverse channel by $P^*_{U|X}$. Next, assume that $\Wrev^*(1|0),\Wrev^*(1,1)$ are on different sides of $q$. Then, there is no loss of generality when assuming that $\Wrev^*(1|0) \leq q$ and that $\Wrev^*(1|1) \geq q$ (if this is not the case, then it can be achieved by reversing the roles of $u=1$ and $u=0$). In this case, the feasibility condition can be rewritten as
    \begin{equation}\label{eqn:feasibleTriangleBH}
        P_U(0) \left(q-\Wrev^*(1|0)\right)+
        P_U(1) \left(\Wrev^*(1|1)-q\right)
        \geq D.
    \end{equation}

    On the other hand, the feasibility condition on $W$ for $R^\triangle(D)$ is
    \begin{equation}
      \EE[\rho(U,Y)]-\EE[\rho(U,X)] \geq D,
    \end{equation}
    which is the same as
    \begin{equation}
      P_U(0)\left(P_Y(1) - \Wrev(1|0)\right) +
      P_U(1)\left(P_Y(0) - \Wrev(0|1)\right) \geq D,
    \end{equation}
    or, equivalently,
    \begin{equation}
      P_U(0)\left(q - \Wrev(1|0)\right) +
      P_U(1)\left(\Wrev(1|1)-q\right) \geq D.
    \end{equation}
    We conclude that the channel $P^*_{U|X}$ is also feasible for the RHS of \eqref{eqn:sameObjDiffFeas}, thereby proving \eqref{eqn:sameObjDiffFeas}, i.e. that $R_\ID(D) = R_\ID^\TCT(D)$.

    If $\Wrev^*(1|0),\Wrev^*(1|1)$ are on the same side of $q$, then $u=0$ and $u=1$ can be merged, following the steps of the merging process at the beginning of the proof. If they are merged, this means that $U$ is no longer a random variable, but a constant, and that $X$ and $U$ are therefore independent. This implies that $R_\ID(D) = 0$, and also that $|p-q| \geq D$. In this special case, we show that $R_\ID^\TCT(D)= 0$ directly.

Our goal is to find a channel $W$ that will make $X,U$ independent, and at the same time be feasible for the minimization of $R_\ID^\TCT(D)$ according to \eqref{eqn:feasibleTriangleBH}. For this purpose, choose the channel $P_{U|X}$ to be
    \begin{equation}
      P_{U|X}(0|x) = \alpha; P_{U|X}(1|x) = 1-\alpha, \mbox{ for any } x\in\{0,1\}.
    \end{equation}
    It is easy to see that $U$ and $X$ are independent, i.e. $I(X;U)=0$, and that $P_U(0) = \alpha$. Next, in order to satisfy \eqref{eqn:feasibleTriangleBH}, choose either $\alpha=0$ or $\alpha = 1$,  according to whether $p<q$ or $q<p$ (if $p=q$, this implies that $D=0$ and then any choice of $\alpha$ will work).
\end{proof}

Theorems~\ref{thm:D0} and \ref{thm:RIDbinHamming}, when combined, result in the following corollary.
\begin{cor}
    If $X\sim\Ber(p)$, $Y\sim\Ber(\tfrac{1}{2})$, and the distortion measure is Hamming, then
    \begin{equation}
        R_\ID(D) = R_\ID^\TCT(D) = R_\ID^\LCT(D) = R(\tfrac{1}{2}-D),
    \end{equation}
    where $R(\cdot)$ is the rate distortion function of the source $X$ and Hamming distortion.
\end{cor}
Note that this result is slightly more general than \cite[Theorem 1]{IdoiaAllerton2013}, since here $X$ is not restricted to be symmetric.

\section{Computing the Identification Rate}\label{sec:computation}
Calculating the value of the achievable rates $R_\ID^\LCT(D)$ and $R_\ID^\TCT(D)$ is relatively easy. The term $D_\ID^\LCT(R)$, shown in \eqref{eqn:DIDLCT}, can be calculated from the distortion rate function (and the achieving reconstruction distribution), which can be calculated with, e.g. the well known Blahut-Arimoto algorithm, or simply as a minimization problem of a linear function over a convex set. The term $D_\ID^\TCT(R)$, shown in \eqref{eqn:DIDTCT}, is also given as a minimization problem of a linear function with convex constraints, and therefore can be solved easily. The general term $R_\ID(D)$, however, is posed as a minimization problem with nonconvex constraints, making its computation a challenge. In this section we give two results that facilitate the computation of this quantity. In Subsection \ref{ssec:card} we improve the bound on the cardinality of $U$, which reduces the dimensions of the optimization problem. In Subsection \ref{ssec:minmin} we describe the process of transforming the (non-convex) problem into a sequence of convex problems that can be solved efficiently.

\subsection{Cardinality of the auxiliary RV $U$}\label{ssec:card}

For the evaluation of $R_\ID(D)$, it was already shown in \cite[Lemma 3]{Ahlswede97} that it suffices to consider only $|\cU|=|\cC|+2$. Here we prove an improvement of the cardinality bound, stated in Theorem~\ref{thm:cardinality}.

\medskip

\begin{proof}[Proof of Theorem~\ref{thm:cardinality}]
We start by proving that taking $|\cU|=|\cX|+1$ suffices to calculate $R_\ID(D)$. The proof follows the idea from \cite{salehi1978cardinality}, i.e. using the strengthened version of Carath\'eodory's theorem due to Fenchel and Eggleston.

Define $|\cX|+1$ functions $\Psi_i: \cP(\cU \ra \cX) \ra \Reals$. In other words, the functions $\Psi_i$ take a conditional distribution from $\cU$ to $\cX$, and return a real number. The functions are given by:
\begin{align}
  \Psi_x(Q) & = Q(x), \mbox{, for } x=1,...,|\cX|-1;\\
  \Psi_{|\cX|}(Q) & = \bar\rho\left(Q,P_Y\right);\\
  \Psi_{|\cX|+1}(Q) & = H(P_X) -H(Q).
\end{align}
Note that in the optimization function $R_\ID(D)$, the objective function can be written as
\begin{equation}
    I(X;U) = \sum_{u \in \cU} P_U(u) \Psi_{|\cX|} (P_{X|U}(\cdot|u)),
\end{equation}
and the constraint can be written as
\begin{equation}
  \sum_{u \in \cU} P_U(u) \Psi_{|\cX|+1} (P_{X|U}(\cdot|u)) \geq D.
\end{equation}
Define the set $\cA$ to be the set of tuples $(\Psi_1(Q),...,\Psi_{|\cX+1|}(Q))$ for all $Q \in \cP(\cX)$. Note that $\cA$ is a closed and connected set, and therefore any point in the convex hull of $\cA$ can be represented as a convex combination of at most $|\cX|+1$ elements of $\cA$ (this is due to the Fenchel-Eggleston-Carath\'eodory theorem, see, e.g. \cite[Theorem 18]{Egglston}). Define $\cB$ to the convex hull of $\cA$. Further, define $\cB_{P_X}$ as
\begin{equation}
  \cB_{P_X} \triangleq \left\{ (\psi_1,...,\psi_{|\cX|+1}: \psi_i = P_X(i), 1\leq i \leq |\cX|-1\right\}.
\end{equation}
In other words, the set $\cB_{P_X}$ contains only convex combinations of $\cA$ that correspond to combinations of distributions on $\cX$ that, when averaged with the convex combination (which represents the distribution on $\cU$), result in the distribution $P_X$.

Finally, let $P^*_{U|X}$ be an achieving distribution for $R_\ID(D)$ and let $P^*_U$ and $P^*_{X|U}$ be the induced marginal on $\cU$ and the reverse conditional distribution, respectively. $P^*_U$ and $P^*_{X|U}$ can be associated with a point in the set $\cB_{P_X}$, where the attained $(R,D)$ pair is given in the last two coordinates of the vector in $\cB_{P_X}$. As claimed before, any point in $\cB$ can be represented as a convex combination of at most $|\cX|+1$ points of $\cA$, and in other words, the same pair $(R,D)$ can be achieved with a distribution that averages only $|\cX|+1$ distributions of $\cX$, i.e. the cardinality of $\cU$ can be limited to $|\cX|+1$.

\medskip
For the second part, recall the following facts:
\begin{itemize}
  \item $R_\ID(D)$ is a convex function of $D$ \cite[Lemma 3]{Ahlswede97}. Denote the region of achievable pairs $(R,D)$ as
\begin{equation}
  \mathcal{R} \triangleq \left\{(R,D) : R \geq R_\ID(D) \right\},
\end{equation}
where%
\footnote{The restriction on the values of $D$ and $R$ is due to the fact that any rate above $\log |\cX|$ is trivially achievable by using the sequence $\bx$ itself as the signature, and for a distortion threshold above $\rho_{\max}$ renders all sequences similar to each other, making the problem degenerate.} $D \in [0,\rho_{\max}]$ and $R \leq \log|\cX|$. Therefore the set is closed, bounded and convex.
\item For a convex set, an \emph{extreme} point is a point in the set that cannot be represented as a nontrivial\footnote{A convex combination is considered trivial if all the coefficients are zero except for one of them (which is equal to one).} convex combination of other points in the set. It is a well-known theorem that any convex set is equal to the convex hull of its extreme points (e.g. \cite[Corr. 18.5.1]{rockafellar70}). For our proof we require the more delicate notion of exposed points.
\item An \emph{exposed} point $p$ of a convex set is a point on the boundary of the set, s.t. there exists a supporting hyperplane of the set at $p$ (a hyperplane that touches the set at $p$, but the set is at one side of the hyperplane), whose intersection with the set itself is equal to $\{p\}$. Any exposed point is also an extreme point. The most useful fact about exposed points is the fact that any closed and bounded convex set is equal to the closure of a convex hull of its exposed points (a special case of \cite[Theorem 18.7]{rockafellar70}). We shall use this fact directly.
\end{itemize}
\medskip

To begin the proof, let $(R_0,D_0)$ (where $R_0 = R_\ID(D_0)$) be an exposed point of the achievable region. Our goal is to show that this point $D_0,R_0$ is achievable with a conditional distribution $P_{U|X}$ s.t. the distribution of $U$ is supported on at most $|\cX|$ elements. Next, let $(c,\lambda)$ be constants s.t. $R = c + \lambda D$ is a supporting line (hyperplane in 2D) of the achievable region at $(D_0,R_0)$, s.t. the intersection of $\mathcal{R}$ and the line contains this point only. Such a line is guaranteed to exist by the assumption that $(R_0,D_0)$ is an exposed point. A typical image is shown in Fig.~\ref{fig:exposed}.
\begin{figure}
  \centering
  \includegraphics[width=4in]{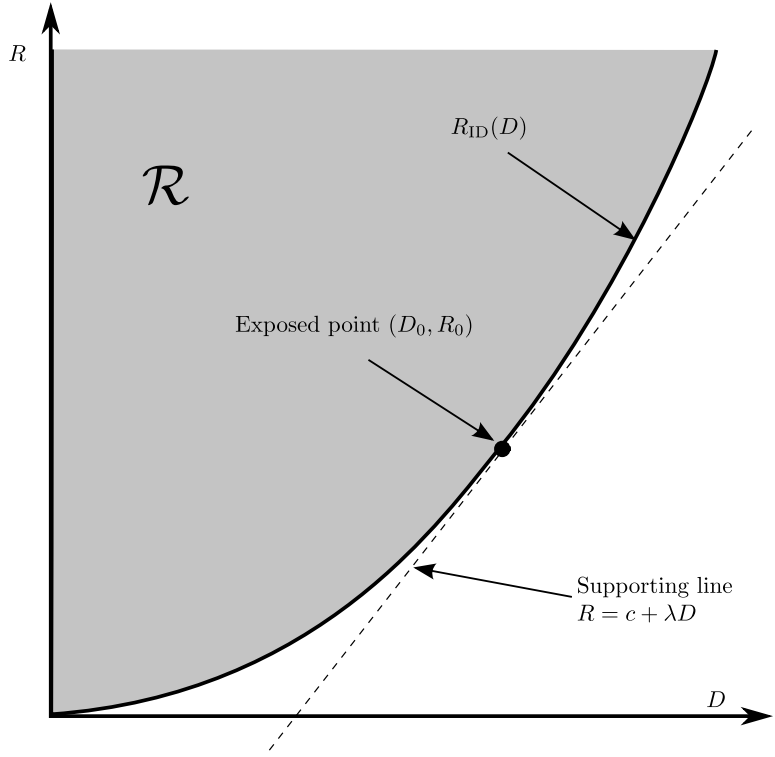}\\
  \caption{The achievable region $\mathcal{R}$, an exposed point $(D,R)$ and a supporting line.}\label{fig:exposed}
\end{figure}
Recall the $R_\ID(D)$ is given by the minimization expression \eqref{eqn:R_ID_info}. Let $P^*_{U|X}$ be an achieving conditional distribution at $D_0$, i.e. that minimizes \eqref{eqn:R_ID_info}. If $R_\ID(D_0) = 0$, this implies that $U$ and $X$ are independent, and therefore $\bar\rho(P_X,P_Y)\geq D$. This means that $R_\ID(D_0)$ can be attained by a trivial distribution of $U$ (where $U$ is a constant). In the general case where $R_\ID(D_0) > 0$, we conclude that the constraint in \eqref{eqn:R_ID_info} is active, and therefore
\begin{equation}
  R_0 = I(X;U); \mbox{ and } D_0 = \EE[\bar\rho(P^*_{X|U}(\cdot,U),P_Y)],
\end{equation}
where $X,U$ are distributed according to $P_X,P^*_{U|X}$.

Next, note that $P^*_{U|X}$ also minimizes the expression
\begin{equation}\label{eqn:Lagrangian}
  I(X;U) - \lambda \EE[\bar\rho(P_{X|U}(\cdot,U),P_Y)],
\end{equation}
where the minimization is without constraints (other than the fact that $P_{U|X}$ is a conditional distribution). This fact follows since an existence of a better minimum would imply a distribution $P'_{U|X}$ for which $(I(X;U),\EE[\bar\rho(P'_{X|U}(\cdot,U),P_Y)])$ falls outside the achievable region (due to the supporting hyperplane property), leading to a contradiction.

Next, claim that from $P^*_{U|X}$, we can construct a distribution $P^{**}_{U|X}$ that attains the same minimum in \eqref{eqn:Lagrangian}, for which the distribution of $P_U$ has at most $|\cX|$ elements. This can be shown by following the same steps as in the first part of the proof (where the cardinality was shown to be bounded by $|\cX|+1$), but now we replace the two functions $\Psi_{|\cX|}$ and $\Psi_{|\cX|+1}$ by a single function $\Psi_{|\cX|}$ that is equal to \eqref{eqn:Lagrangian}.

Since the new distribution $P^{**}_{U|X}$ attains the same minimum in \eqref{eqn:Lagrangian} as $P^*_{U|X}$ does, we conclude that the point $(D_1,R_1)$, given by
\begin{equation}
  R_1 = I(X;U); \mbox{ and } D_1 = \EE[\bar\rho(P^*_{X|U}(\cdot,U),P_Y)],
\end{equation}
where $X,U$ are distributed according to $P_X,P^{**}_{U|X}$, satisfies the same line equation $R_1 = c + \lambda D_1$. However, since we assumed that $(D_0,R_0)$ is an \emph{exposed} point of the achievable region, then by definition
\begin{equation}
  (D_1,R_1) = (D_0,R_0).
\end{equation}
In other words, the distribution $P^{**}_{U|X}$ attains the minimum of the original minimization problem \eqref{eqn:R_ID_info}.

The proof is concluded since, as noted before, a bounded, closed and convex set is equal to the closure of the convex hull of its exposed points. As a result, the achievable region $\mathcal{R}$ can be calculated by calculating $R^k_\ID(D)$, and then taking the lower convex envelope (the closure operation has no practical effect).
\end{proof}

\subsection{Conversion to a set of convex functions}\label{ssec:minmin}

Consider first the case where the distortion measure is Hamming, and $P_X,P_Y$ are arbitrary distributions on $\cX$. In this case, it is not hard to verify that
\begin{align}
  \bar\rho(P_X,P_Y)
    &= \frac{1}{2}\|P_X-P_Y\|_1 \label{eqn:first}\\
    &= \frac{1}{2}\sum_{x \in \cX}|P_X(x)-P_Y(x)|.
\end{align}

With this fact, we can rewrite the identification rate as
\begin{equation}
  \min_{P_{U|X}} I(X;U),
\end{equation}
where the minimization is w.r.t. all conditional distributions $P_{U|X}$ s.t.
\begin{equation}\label{eqn:HammingConditionL1}
  \sum_{u \in \cU} P_U(u) \sum_{x \in \cX} \left|P_{X|U}(x|u) - P_Y(x) \right| \geq 2D.
\end{equation}

Define $\cF$ to be the set of all functions that take a pair in $\cX \times \cU$ and return a binary value. With this, the condition \eqref{eqn:HammingConditionL1} is equivalent to
\begin{equation}\label{eqn:HammingConditionMultiple}
  \max_{f\in \cF} \left[\sum_{u \in \cU} P_U(u) \sum_{x \in \cX} (-1)^{f(x,u)}\left(P_{X|U}(x|u) - P_Y(x) \right)\right] \geq 2D.
\end{equation}
Alternatively, we may require $P_{U|X}$ to satisfy
\begin{equation}\label{eqn:HammingConditionMultiple2}
  \sum_{u \in \cU} P_U(u) \sum_{x \in \cX} (-1)^{f(x,u)}\left(P_{X|U}(x|u) - P_Y(x) \right) \geq 2D,
\end{equation}
for \emph{some} function $f\in \cF$. Define the LHS of \eqref{eqn:HammingConditionMultiple2} as $L_f(P_{U|X})$, and rewrite it as
\begin{align}
L_f(P_{U|X}) &\triangleq \sum_{u \in \cU} P_U(u) \sum_{x \in \cX} (-1)^{f(x,u)}\left(P_{X|U}(x|u) - P_Y(x) \right) \label{eqn:LfFirst}\\
&= \sum_{u \in \cU}\sum_{x \in \cX} (-1)^{f(x,u)} \left(P_{U|X}(u|x) P_X(x)  - P_U(u) P_Y(x) \right) \\
&= \sum_{u \in \cU}\sum_{x \in \cX} (-1)^{f(x,u)} \left(P_{U|X}(u|x) P_X(x)  - P_U(u) P_Y(x) \right),\label{eqn:Lf}
\end{align}
which shows that $L_f(P_{U|X})$ is a \emph{linear} function of the optimization variable $P_{U|X}$. Finally, we can rewrite the optimization problem as
\begin{equation}\label{eqn:minmin}
  R_\ID(D) = \min_{f\in\cF} \min_{L_f(P_{U|X}) \geq 2D} I(X;U).
\end{equation}
The expression \eqref{eqn:minmin} gives rise to the following scheme for computing $R_\ID(D)$: since $L_f(\cdot)$ is a linear function, the inner optimization in \eqref{eqn:minmin} is that of a convex target function with linear constraints, and can be solved efficiently (e.g. via \texttt{cvx} \cite{cvx}). To get the value of $R_\ID(D)$, simply repeat the inner optimization for all $f\in\cF$, and take the overall minimal value.

The main problem with this approach is that the number of optimization problems can be very high. Assume that $\cU = \cX$, following the previous subsection. The size of $\cF$ is $2^{|\cX|^2}$. For $|\cX|=5$, for example, one would need to solve $2^{25} \cong 33.5\times 10^6$ optimization problems. We can slightly improve the situation by utilizing symmetries in the expression \eqref{eqn:HammingConditionMultiple2}.

\begin{theorem}\label{thm:minmin}
Define the set $\cF' \subseteq \cF$ as follows. The set $\cF'$ shall contain only functions $f(\cdot,\cdot)$ where:
\begin{itemize}
  \item For all $u$, the function $f(\cdot,u)$ cannot be constant (in $x$). In other words, for the inner summation in \eqref{eqn:HammingConditionMultiple2}, some of the summands must be flipped and some not.
  \item There are no $u_1\neq u_2 \in \cU$ s.t. $\forall_{x\in\cX} f(x,u_1) = f(x,u_2)$.
\end{itemize}
    Then in the double optimization of the form \eqref{eqn:minmin}, it suffices to consider functions $f\in \cF'$ as defined above. In addition, the number of such functions is given by
\begin{equation}\label{eqn:F'size}
    |\cF'| = \comb{2^{|\cX|}-2}{|\cU|}.
\end{equation}
\end{theorem}
\begin{proof}Appendix~\ref{app:minmin}.
\end{proof}
For quick reference, we show in Table~\ref{tab:optHamming} the improvement in the number of optimization problems that is sufficient to solve as a result of Theorem~\ref{thm:minmin}.
\begin{table}
  \centering
    \begin{tabular}{ | l | l | l |}
    \hline
    $|\cX|$ & $|\cF|$ & $|\cF'|$  \\ \hline
    $2$ & $16$ & $1$  \\ \hline
    $3$ & $512$ & $20$  \\ \hline
    $4$ & $65536$ & $1001$  \\ \hline
    $5$ & $34\cdot 10^6$ & $142 \cdot 10^3$  \\ \hline\hline
    $10$ & $1.3 \cdot 10^{30}$ & $3.3 \cdot 10^{23}$  \\ \hline
    \end{tabular}
  \caption{Number of convex optimization problems to be solved in order to calculate  $R_\ID(D)$.}\label{tab:optHamming}
\end{table}
For example, the identification rate in Fig.~\ref{fig:ternary} above, for ternary alphabet, was calculated using the method above. At each point, we have solved $20$ convex optimization programs and took the minimum value.

\medskip

As seen in Table~\ref{tab:optHamming}, the proposed method for the computation of $R_\ID(D)$, although it improves on the naive \eqref{eqn:minmin}, is only effective for small values of $|\cX|$. It is therefore an open problem how to calculate $R_\ID(D)$ effectively for larger alphabets.

\medskip

Finally, we briefly note how to extend the process described here for arbitrary distortion measures. The key fact in the computation of $R_\ID(D)$ in the Hamming case is the fact that the function $f(P)= \bar\rho(P,P_Y)$ can be represented as a maximum of \emph{linear} functions of $P$. This fact holds in the general case as well, following the fact that the \emph{epigraph} of the function $f$, defined as \begin{equation}
  \mathrm{epi} f \triangleq \left\{(P,D) \in \cP(\cX)\times \Reals : f(P) \leq D\right\},
\end{equation}
is always a polyhedron. For the proof of this fact, see \cite{IDexponent}. Once the $\bar\rho$-distance has been represented as a maximization of linear functions, the process described in Equations \eqref{eqn:first}-\eqref{eqn:minmin} can be followed. Note that as in the Hamming case, it is expected that this approach will only allow easy computation of $R_\ID(D)$ in cases where the alphabet $\cX$ is small.

\section{Conclusion and Future Work}\label{sec:Summary}

In this paper we have established the fundamental limit of compression for similarity identification: the minimal compression rate that allows reliable answers to the query ``is the compressed sequence similar to the query sequence''. While the achievability part was mostly derived in previous work (namely Ahlswede et al. \cite{Ahlswede97}), for the converse part we combined the approach of \cite{Ahlswede97} with the blowing-up lemma. We then investigated the achievable performance when using lossy compression as a building block, and provided a method for efficiently computing $R_\ID(D)$ for small alphabets.

There are several directions for future research that are natural given the result in the paper, some are theoretical, and some are more on the practical side. Future work that relates to theory includes the following:
\begin{itemize}
  \item Symmetric compression schemes: how does $R_\ID(D)$ change when the query sequence is also compressed, and possibly at a different rate than the source sequence? While the achievability part of this question is rather similar in spirit to the one presented here, the converse seems to be more complicated.
  \item Characterization of the ``identification exponent'' -- how fast does the probability of $\MAYBE$ (or, similarly, the false-positive probability) go to zero when the sequence length grows? Results for variable length compression has been presented in \cite{Ahlswede97}. However, they depend on an auxiliary random variable with unbounded cardinality, making the result uncomputable. Recently \cite{IDexponent}, the cardinality issue has been resolved, along with the exponent for the fixed-length compression case (which is different than that of the variable length).
  \item In addition to the error exponent, in lossy source coding (and also in channel coding) there exist additional results that characterize the tradeoff between rate, reliability and sequence length. Such results include different asymptotics (i.e. ``dispersion''-type results \cite{IngberKochmanSourceDispersion2011}) and also explicit bounds for finite sequence length (e.g. \cite{KostinaVerduFiniteBlock}). It would be interesting to discover similar results for the setting of the current paper.
  \item More complicated source and query models: how do the results change when the source and/or query sequence are no longer i.i.d.? For finite-order Markov-type sources, it is expected that an approach based on the method of types (namely its extension for sources with memory, e.g. \cite[Sec. VII]{csiszar1998method}) will lead to the right direction. For the case where the source and query sequences are statistically dependent, Ahlswede et al. \cite{Ahlswede97} provide partial results, as the dependent case seems to be more difficult.
\end{itemize}

\medskip

On the practical side, these are possible directions for future research, some of which are already being pursued:
\begin{itemize}
  \item Practical schemes for compression for similarity queries: Shannon's classical rate distortion theorem \cite{shannon1959coding} is now over 50 years old, but practical schemes for approaching the rate-distortion limit has only appeared roughly in the last two decades. It would be interesting to study how to harness the vast amount of work that has been done on practical source coding systems for the related (but different) task of compression for identification. This direction is already being pursued, with preliminary results reported in \cite{IdoiaAllerton2013} and \cite{idoiaDCC2014}.
  \item Computation of $R_\ID(D)$: As discussed in Section~\ref{sec:computation}, the computation of $R_\ID(D)$ is a challenge, mainly due to the fact that it is given as a non-convex optimization problem. While for the Hamming case and an alphabet of small size we have presented an efficient way to calculate $R_\ID(D)$, the general problem remains open. It would be interesting to study other approaches, perhaps in the spirit of the well known Blahut-Arimoto iterative algorithm, for efficiently computing $R_\ID(D)$.
\end{itemize}

\section*{Acknowledgement}
The authors would like to thank Thomas Courtade for fruitful discussions and for proving an earlier version of Theorem~\ref{thm:UniversalLB_Hamming}, and to Golan Yona for introducing us to the world of biological databases that provided the initial motivation for this work.

\appendices

\section{Equivalence of Fixed and Variable Length Identification Rate}\label{app:VL=FL}

\begin{proof}[Proof of Prop.~\ref{prop:VL=FL}]
    Our goal is to show that $R_\ID(D) \leq R_\ID^{vl}(D)$. Let $T^{(i)}_{vl},g^{(i)}_{vl}$ be a sequence of variable length schemes of rate $R$, that achieve a vanishing probability for $\MAYBE$. We will construct a sequence of \emph{fixed-length} schemes with rate arbitrarily close to $R$, that also attain a vanishing probability for $\MAYBE$.

The fixed-length scheme shall be constructed as a concatenation of $M$ variable-length schemes, operating on a single sequence $\bx$ of length $nM$ (each instance of the variable length scheme operates on a separate block in the input sequence).

Define $L_m$ to be normalized length of the binary codeword of the $m$-th block of $\bX$, i.e. the output of the $m$-th instance of the variable length scheme. Note that since the compressed sequence $\bX$ is i.i.d., and the random variables $L_m$ are i.i.d. as well.

Let $\eps>0$, $\Delta R>0$ to be arbitrarily small constants.
By a standard Chebyshev-type argument, we have that
\begin{equation}
  \Pr\left\{\frac{1}{M}\sum_{m=1}^M L_m > R + \Delta R\right\} \leq \frac{R^2}{M\cdot \Delta R^2}.
\end{equation}
Choose $M$ s.t. the RHS in the above inequality is equal to $\eps/2$.

The new fixed-length scheme shall work as follows:

\textbf{Encoding}:
\begin{itemize}
  \item Encode each sub-block (of length $n$) with the underlying variable length scheme.
  \item Calculate $\frac{1}{M}\sum_m L_m$. If larger than $R+\Delta R$, set the signature of the entire sequence to $\texttt{e}$, denoting ``erasure''.
  \item Otherwise, the signature is the concatenation of the variable-length codewords that correspond to each of the sub-block. Note that is guaranteed that the number of different signatures is at most $2^{nM (R+\Delta R)}+1$, i.e. the rate is arbitrarily close to $R$ (the added one is due to the erasure symbol).
\end{itemize}
\textbf{Decision function}:\\
Given a signature and a query sequence $\by$, the decision function $g(\cdot,\cdot)$ is defined as follows.
\begin{itemize}
  \item If the signature equal to $\texttt{e}$, answer $\MAYBE$.
  \item Otherwise, compute the answers of the $M$ sub-schemes where the input for each of them is the $m$-th signature (corresponding to the $m$-th block of $\bx$), and the $m$-th block of $\by$.
  \item Finally, answer $\NO$ if \emph{all} the sub-schemes returned a $\NO$. Otherwise, return a $\MAYBE$.
\end{itemize}

\textbf{Analysis:}\\
The probability of $\MAYBE$ in the overall scheme can be bounded by, by the union bound, as
\begin{equation}
  \Pr\{\MAYBE\} \leq \Pr\{\mbox{The signature is } \texttt{e}\} + M \times \Pr\{\mbox{A sub-scheme has returned a } \MAYBE\}.
\end{equation}
Recall that $M$ was chosen to be equal to $c/\eps$, where $c$ is a constant (independent of $n$), and that the probability of erasure is bounded by $\eps/2$. Overall, we have
\begin{equation}
  \Pr\{\MAYBE\} \leq \eps/2 + \frac{c}{\eps} \times \Pr\{\mbox{A sub-scheme has returned a } \MAYBE\}.
\end{equation}
Finally, note that the probability of $\MAYBE$ in the underlying scheme can be made to be arbitrarily small (while letting $n$ grow), and specifically, it can be made smaller than $\eps^2/(2c)$. With this choice, the overall probability of $\MAYBE$ is upper bounded by $\eps$, which was chosen to be arbitrarily small. Since the rate of the fixed-length scheme is arbitrarily close to $R$, this completes the proof of the proposition.
\end{proof}

\section{}\label{app:UniversalLB_Hamming}
\begin{proof}[Proof of Theorem~\ref{thm:UniversalLB_Hamming}]
Let $\rho(\cdot,\cdot)$ denote the Hamming distortion. The proof relies on a bound on the $\bar\rho$ distance due to Marton \cite{marton1996bounding}, where it is called a $\bar{d}$-distance in the context of Hamming distance only. The result in \cite[Prop.~1]{marton1996bounding} says that for any two distribution $(P_A,P_B)$,
\begin{equation}\label{eqn:pinsker}
  \bar\rho(P_A,P_B) \leq \left[\frac{1}{2}D_e(P_A||P_B)\right]^{1/2},
\end{equation}
where $D_e(\cdot||\cdot)$ is the KL divergence, given in nats\footnote{Note that since the $\bar \rho$ distance for the Hamming distance is equal to the $\ell_1$ distance between the distributions, \eqref{eqn:pinsker} is nothing but Pinsker's inequality.}. With this result, consider the constraint on $P_{XU}$ in the expression for $R_\ID(D)$:
\begin{align}
  \sum_u P_U(u) \bar\rho(P_{X|U}(\cdot|u),P_Y)
&\leq \sum_u P_U(u) \left[\frac{1}{2}D_e(P_{X|U}(\cdot|u)||P_Y)\right]^{1/2}   \\
&\leq \left[\frac{1}{2}\sum_u P_U(u) D_e(P_{X|U}(\cdot|u)||P_Y)\right]^{1/2},
\end{align}
where the second inequality follows from Jensen's inequality and the concavity of $\sqrt{\cdot}$. By writing the explicit expression for the divergence, we obtain:
\begin{align}
  \sum_u P_U(u) D_e(P_{X|U}(\cdot|u)||P_Y)
&= \frac{1}{\log e}\sum_u P_U(u) D(P_{X|U}(\cdot|u)||P_Y)\\
&= \frac{1}{\log e}\sum_u P_U(u) \sum_x P_{X|U}(x|u) \log \frac{P_{X|U}(x|u)}{P_{Y}(x)}\\
&= \frac{1}{\log e}\sum_u P_U(u) \sum_x P_{X|U}(x|u) \log \frac{P_{X|U}(x|u)P_X(x)}{P_X(x)P_{Y}(x)}\\
&= \frac{I(X;U) + D(P_X||P_Y)}{\log e}.
\end{align}
Therefore the constraint
\begin{equation}
  \sqrt{\frac{I(X;U)+D(P_X||P_Y)}{2 \log e}} \geq D
\end{equation}
is more loose than that of the identification rate, and therefore we obtain
\begin{align}
  R_\ID(D) & \geq \min_{I(X;U) + D(P_X||P_Y) \geq 2D^2\log e} I(X;U) \\
   & \geq 2D^2\log e - D(P_X||P_Y).
\end{align}
since $R_\ID(D)$ is nonnegative, the proof is concluded.
\end{proof}

\section{}\label{app:MostAiAreLarge}

\begin{proof}[Proof of Lemma~\ref{lem:MostAiAreLarge}]
We first give an upper bound on $\Pr\{\bX \in A_i\}$ in terms of $|A_i|$.
\begin{align}
    \Pr\{\bX \in A_i\}
    &= \sum_{P \in \cP_n(\cX) :\|P-P_X\|_\infty \leq \gamma} \Pr\{\bX \in A_i \cap T_P\} \\
    & = \sum_{P \in \cP_n(\cX) :\|P-P_X\|_\infty \leq \gamma} \frac{|A_i \cap T_P|}{|T_P|}\\
   & \leq \sum_{P \in \cP_n(\cX) :\|P-P_X\|_\infty \leq \gamma} \frac{|A_i \cap T_P|}{\tfrac{1}{(n+1)^{|\cX|}}2^{nH(P)}} \\
   & = \sum_{P \in \cP_n(\cX) :\|P-P_X\|_\infty \leq \gamma}  \frac{|A_i \cap T_P|}{2^{nH(P) - |\cX| \log(n+1)}}\\
   & = \sum_{P \in \cP_n(\cX) :\|P-P_X\|_\infty \leq \gamma}  \frac{|A_i \cap T_P|}{2^{nH(P) - |\cX| \log(n+1)}}\\
   & = \sum_{P \in \cP_n(\cX) :\|P-P_X\|_\infty \leq \gamma} |A_i \cap T_P| 2^{-nH(P) + |\cX| \log(n+1)}\\
   & \leq \sum_{P \in \cP_n(\cX) :\|P-P_X\|_\infty \leq \gamma}
   |A_i \cap T_P| 2^{-n[H(P_X) - \gamma |\cX|\log (1/\gamma)]+ |\cX| \log(n+1)}\\
   & \leq \sum_{P \in \cP_n(\cX) :\|P-P_X\|_\infty \leq \gamma} |A_i \cap T_P| 2^{-n[H(P_X) - \gamma']}   \\
   & = 2^{-n[H(P_X) - \gamma']} \sum_{P \in \cP_n(\cX) :\|P-P_X\|_\infty \leq \gamma}
   |A_i \cap T_P| \\
   & = |A_i|\cdot 2^{-n(H(P_X) - \gamma')}\label{eqn:PrAiUpperBound}.
\end{align}
The first two inequalities in the above derivation follow from \cite[Lemma 2.3]{CsiszarKorner2nd} and  \cite[Lemma 2.7]{CsiszarKorner2nd} respectively. The last inequality follows from the definition of $\gamma'$ and by setting $n_0$ to be the smallest $n_0$ s.t.
$\frac{1}{n} \log(n+1) \leq \gamma \log (1/\gamma)$.

Next, we have (for any $R'$):
\begin{align}
    \sum_{i: |A_i| \leq 2^{nR'} }\Pr\{\bX \in A_i\}
    &\overset{(a)}\leq \sum_{i: |A_i| \leq 2^{nR'} } |A_i| \cdot 2^{n(H(P_X)-\gamma')}\\
    &\leq \sum_{i: |A_i| \leq 2^{nR'} } 2^{nR'} \cdot 2^{n(H(P_X)-\gamma')}\\
    &\overset{(b)}\leq 2^{n(R' + R - H(P_X) + \gamma')},
\end{align}
where $(a)$ follows from \eqref{eqn:PrAiUpperBound} and $(b)$ follows since the sum contains $2^{nR}$ elements. The proof of the lemma is concluded by choosing $R'= H(P_X)-R -2\gamma'$.
\end{proof}

\section{}\label{app:LogSumLemma}
\begin{proof}[Proof of Lemma~\ref{lem:LogSumLemma}]
    We start with
    \begin{align}
      H(\tilde\bY)
       & =\sum_{i=1}^n H(\tilde Y_i | \tilde Y^{i-1}) \\
       & \geq \sum_{i=1}^n H(\tilde Y_i | \tilde Y^{i-1} \tilde X^{i-1} \tilde U^{i-1}) \\
       & \overset{(a)}= \sum_{i=1}^n H(\tilde Y_i | \tilde X^{i-1} \tilde U^{i-1}) \\
       & \overset{(b)}= \sum_{i=1}^n H(\tilde Y_i | \tilde X^{i-1})\\
       & = \sum_{i=1}^n \sum_{x^{i-1}\in A_{i-1}} \Pr\left\{\tilde X^{i-1}=x^{i-1}\right\}
      H(\tilde Y_i | \tilde X^{i-1}=x^{i-1})\\
      & = \sum_{i=1}^n \sum_{x^{i-1}\in A_{i-1}} \Pr\left\{\tilde X^{i-1}=x^{i-1}\right\}\nonumber\\
      &\quad \times \sum_{y\in\cY} \Pr\{\tilde Y_i=y | \tilde X^{i-1}=x^{i-1}\} \log \frac{1}{\Pr\{\tilde Y_i=y | \tilde X^{i-1}=x^{i-1}\}},\label{eqn:H_Y_tilde}
\end{align}    where $(a)$ follows since $\tilde Y^i - (\tilde X^{i-1},\tilde U^{i-1}) - \tilde Y^{i-1}$ form a Markov chain, and $(b)$ follows since $\tilde U^{i-1}$ is a function of $\tilde X^{i-2}$.

Next, it is not hard to verify that
\begin{align}
  \frac{1}{n} \sum_{i=1}^n \Pr\{\tilde X_i = x\} &= P_{X_0}(x),\\
  \frac{1}{n} \sum_{i=1}^n \Pr\{\tilde Y_i = y, \tilde U_i = u \}
  &= \Pr\{Y_0 = y, U_0 = u\}\\
  &= \sum_{x \in \cX} Q(x,u) V(y | x,u).
\end{align}
We remind that $X_0,U_0,Y_0$ are random variables that are distributed according to $(Q,V)$. i.e. that \begin{align}
  \Pr\{X_0 = x,U_0 = u\} &= P_{X_0U_0}(x,u) = Q(x,u),\\
  \Pr\{Y_0 = y|X_0=x,U_0 = u\} &= P_{Y_0|X_0U_0}(y|x,u) = V(y|x,u).
\end{align}

Next, write:
\begin{align}
  n H(Y_0|U_0)
  & = n \sum_{u\in\cU_m,y\in\cY} \Pr\{Y_0=y,U_0=u\} \log \frac{1}{P_{Y_0|U_0}(y|u)}\label{eqn:ref1}\\
  & = \sum_{i=1}^n \sum_{u\in\cU_m,y\in\cY}  \Pr\{\tilde Y_i = y, \tilde U_i = u \}
  \log \frac{1}{P_{Y_0|U_0}(y|u)}\\
  & = \sum_{i=1}^n \EE
\left[ \log \frac{1}{P_{Y_0|U_0}(\tilde Y_i|\tilde U_i)}\right]\\
  & = \sum_{i=1}^n \EE
\left[ \log \frac{1}{
P_{Y_0|U_0}\left(\tilde Y_i|\phi \left(\tilde X^{i-1}\right)\right)
}\right]\\
  & = \sum_{i=1}^n \sum_{x^{i-1}\in A_{i-1}}  \Pr\{\tilde X^{i-1} = x^{i-1}\}\nonumber\\
  &\quad \times \sum_{y \in \cY} \Pr\{\tilde Y_i = y | \tilde X^{i-1} = x^{i-1}\}
\log \frac{1}{P_{Y_0|U_0}\left(y|\phi \left(x^{i-1}\right)\right)}\label{eqn:ref2}.
\end{align}
Combined with \eqref{eqn:H_Y_tilde} we can write
\begin{align}
   n H(Y_0|U_0) - H(\tilde\bY)
   & \leq \sum_{i=1}^n \sum_{x^{i-1}\in A_{i-1}}  \Pr\{\tilde X^{i-1} = x^{i-1}\}\nonumber\\
  &\quad \times \sum_{y \in \cY} \Pr\{\tilde Y_i = y | \tilde X^{i-1} = x^{i-1}\}
\log \frac{\Pr\{\tilde Y_i=y | \tilde X^{i-1}=x^{i-1}\}}{
P_{Y_0|U_0}\left(y|\phi \left(x^{i-1}\right)\right)}.\label{eqn:preLogSum}
\end{align}
Next, note that

\begin{align}
P_{Y_0|U_0}\left(y|\phi \left(x^{i-1}\right)\right)
  &= \sum_{x \in \cX} P_{X_0|U_0}(x|\phi \left(x^{i-1}\right)) V(y | x,\phi(x^{i-1})), \label{eqn:LogSumIngredient1}\\
  \Pr\left\{\tilde Y_i=y | \tilde X^{i-1}=x^{i-1}\right\} & =
  \sum_{x \in \cX} \Pr\left\{\tilde X_i = x |\tilde X^{i-1}=x^{i-1} \right\} V(y | x,\phi(x^{i-1})).\label{eqn:LogSumIngredient2}
\end{align}
With \eqref{eqn:LogSumIngredient1} and \eqref{eqn:LogSumIngredient2}, we can use the log-sum inequality and write:
\begin{align}
  & \sum_{y \in \cY} \Pr\{\tilde Y_i = y | \tilde X^{i-1} = x^{i-1}\}
\log \frac{\Pr\{\tilde Y_i=y | \tilde X^{i-1}=x^{i-1}\}}{
P_{Y_0|U_0}\left(y|\phi \left(x^{i-1}\right)\right)}\nonumber\\
  & \leq \sum_{y \in \cY} \sum_{x \in \cX}
  \Pr\left\{\tilde X_i = x |\tilde X^{i-1}=x^{i-1} \right\} V(y | x,\phi(x^{i-1}))
   \log \frac
   {\Pr\left\{\tilde X_i = x |\tilde X^{i-1}=x^{i-1} \right\}}
   {P_{X_0|U_0}(x|\phi \left(x^{i-1}\right))}\\
  & =  \sum_{x \in \cX}
  \Pr\left\{\tilde X_i = x |\tilde X^{i-1}=x^{i-1} \right\}\log \frac
   {\Pr\left\{\tilde X_i = x |\tilde X^{i-1}=x^{i-1} \right\}}
   {P_{X_0|U_0}(x|\phi \left(x^{i-1}\right))}.
\end{align}
Combined with \eqref{eqn:preLogSum} we have
\begin{align}
   n H(Y_0|U_0) - H(\tilde\bY)
   & \leq \sum_{i=1}^n \sum_{x^{i-1}\in A_{i-1}}  \Pr\{\tilde X^{i-1} = x^{i-1}\}\nonumber\\
  &\quad \times \sum_{x \in \cX}
  \Pr\left\{\tilde X_i = x |\tilde X^{i-1}=x^{i-1} \right\}\log \frac
   {\Pr\left\{\tilde X_i = x |\tilde X^{i-1}=x^{i-1} \right\}}
   {P_{X_0|U_0}(x|\phi \left(x^{i-1}\right))}\label{eqn:similarly}\\
   & = nH(X_0|U_0)-H(\tilde\bX)\\
   & = nH(X_0|U_0)-\log |\tilde A|\\
   & \leq n \frac{\log^2m}{m},
\end{align}
where \eqref{eqn:similarly} follows from derivations similar to \eqref{eqn:ref1}-\eqref{eqn:ref2}, and the last inequality follows from \eqref{eqn:ITSproprtey3}. This concludes the proof of the lemma.
\end{proof}

\section{}\label{app:minmin}
\begin{proof}[Proof of Theorem~\ref{thm:minmin}]
Our goal is to show that
\begin{equation}\label{eqn:F=F'}
  \min_{f\in\cF} \min_{L_f(P_{U|X}) \geq 2D} I(X;U) = \min_{f\in\cF'} \min_{L_f(P_{U|X}) \geq 2D} I(X;U).
\end{equation}
Let $P^*_{U|X}$ be a minimizer of the LHS of \eqref{eqn:F=F'}. Our goal is to construct some $f^*\in\cF'$ s.t. $P^*_{U|X}$ will be a minimizer of $\min_{L_{f^*}(P_{U|X}) \geq 2D} I(X;U)$.
To this end, define $f^*(x,u)$ as follows:
\begin{equation}
    f^*(x,u) = \left\{
                 \begin{array}{ll}
                   0, & \hbox{if $P_{U|X}^*(u|x) P_X(x) > P^*_U(u) P_Y(x)$;} \\
                   1, & \hbox{if $P_{U|X}^*(u|x) P_X(x) < P^*_U(u) P_Y(x)$.}
                 \end{array}
               \right.
\end{equation}
where $P_U^*$ is the marginal of $U$ that results from $P_X,P_{U|X}^*$. Whenever $P_{U|X}^*(u|x) P_X(x) = P^*_U(u) P_Y(x)$, break ties arbitrarily so that $f^* \in \cF'$, e.g. by setting $f^*(u,x) = 0$ if $x=0$, and $1$ otherwise. With this definition, it is obvious that
\begin{align}
L_{f^*}(P^*_{U|X})
&= \sum_{u \in \cU}\sum_{x \in \cX} \left|P_{U|X}(u|x) P_X(x)  - P_U(u) P_Y(x) \right|,
\end{align}
and also that for all $f \in \cF$, $L_{f^*}(P^*_{U|X}) \geq L_{f}(P^*_{U|X})$. The conclusion is that if $P^*_{U|X}$ is feasible for some $f \in \cF$, (i.e. satisfies $L_f(P^*_{U|X}) \geq 2D$, then for sure it will also be feasible for $f^*$, and hence we have equality in \eqref{eqn:F=F'}.

\medskip

For the second part, it is convenient to consider function $f\in \cF$ as a matrix, with $f(u,x_1),f(u,x_2),...$ in the $u$'th row. The first claim of the theorem says that there cannot be any rows in the matrix with fixed values. In other words, there are $2^{|\cX|}-2$ possible values for every row in the matrix. The second claim is that there are no two rows that are equal to each other. This immediately shows that the number of such matrices, which is equal to $|\cF'|$, is simply the number of combinations of $|\cU|$ different rows from $2^{|\cX|}-2$ possible values. The order of the rows does not matter, since it is equivalent to relabeling the values $u \in \cU$, and hence we arrive at \eqref{eqn:F'size}. To show that indeed there is no need to have repeated rows in the matrix, note the following.

Suppose that $P_{U|X}^*$ is a minimizer for the LHS of \eqref{eqn:F=F'}, at some $f$ with two equal rows in the matrix corresponding to $f$. Denote by $P_{X|U}^*$ the reverse channel. Let $u_1,u_2$ correspond to the two identical rows in the matrix. Then, construct the following conditional distribution $P^{**}_{U|X}$, by \emph{merging} the outputs $u_1$ and $u_2$ into a new symbol $u'$ (a process similar to that of Theorem~\ref{thm:RIDbinHamming}). The new distribution results in mutual information $I(P_X,P^{**}_{U|X})$ that is not smaller than the one obtained by $P^{*}_{U|X}$, because of the data processing inequality. Next, rename $u'$ to be $u_1$, and add a fictitious new symbol $u_2$ with probability zero. Then, define a new function $f'(\cdot,\cdot)$ to be equal to $f(\cdot,\cdot)$ for all $u \neq u_2$, and for $u_2$, choose $f'(u_2,\cdot)$ to be a new line in the matrix that hasn't occurred there. This is guaranteed to exist, since there are $2^{|\cX|}-2 > |\cU|-1$ such possible values. Finally, note that by construction,
\begin{equation}
  L_f(P^{*}_{U|X}) = L_{f'}(P^{**}_{U|X}).
\end{equation}
In other words, if $P^{*}_{U|X}$ is feasible for some $f\in \cF$, then there exists a distribution $P^{**}_{U|X}$, with better (lower) mutual information, that is feasible for another $f'\in\cF'$.

The conclusion is, then, that it suffices to consider only functions $f \in \cF'$.
\end{proof}

\bibliographystyle{IEEEtran}
\bibliography{Master}

\begin{thebibliography}{10}
\providecommand{\url}[1]{#1}
\csname url@samestyle\endcsname
\providecommand{\newblock}{\relax}
\providecommand{\bibinfo}[2]{#2}
\providecommand{\BIBentrySTDinterwordspacing}{\spaceskip=0pt\relax}
\providecommand{\BIBentryALTinterwordstretchfactor}{4}
\providecommand{\BIBentryALTinterwordspacing}{\spaceskip=\fontdimen2\font plus
\BIBentryALTinterwordstretchfactor\fontdimen3\font minus
  \fontdimen4\font\relax}
\providecommand{\BIBforeignlanguage}[2]{{%
\expandafter\ifx\csname l@#1\endcsname\relax
\typeout{** WARNING: IEEEtran.bst: No hyphenation pattern has been}%
\typeout{** loaded for the language `#1'. Using the pattern for}%
\typeout{** the default language instead.}%
\else
\language=\csname l@#1\endcsname
\fi
#2}}
\providecommand{\BIBdecl}{\relax}
\BIBdecl

\bibitem{Ahlswede97}
R.~Ahlswede, E.-h. Yang, and Z.~Zhang, ``Identification via compressed data,''
  \emph{IEEE Trans. on Information Theory}, vol.~43, no.~1, pp. 48 --70, Jan
  1997.

\bibitem{ICW13_DCC}
A.~Ingber, T.~Courtade, and T.~Weissman, ``Quadratic similarity queries on
  compressed data,'' in \emph{Data Compression Conference (DCC)}, 2013, pp.
  441--450.

\bibitem{ICW13_IT}
\BIBentryALTinterwordspacing
A.~Ingber, T.~A. Courtade, and T.~Weissman, ``Compression for quadratic
  similarity queries,'' Submitted to IEEE Trans. on Information Theory, 2013.
  [Online]. Available: \url{http://arxiv.org/abs/1307.6609}
\BIBentrySTDinterwordspacing

\bibitem{IDexponent}
A.~Ingber and T.~Weissman, ``The error exponent in compression for similarity
  identification,'' in prep.

\bibitem{IDexponentExact}
A.~Ingber, T.~Courtade, and T.~Weissman, ``Compression for exact match
  identification,'' in \emph{IEEE International Symposium on Information Theory
  Proceedings (ISIT)}, 2013, pp. 654--658.

\bibitem{tuncel2004rate}
E.~Tuncel, P.~Koulgi, and K.~Rose, ``Rate-distortion approach to databases:
  Storage and content-based retrieval,'' \emph{IEEE Trans. on Information
  Theory}, vol.~50, no.~6, pp. 953--967, 2004.

\bibitem{WynerZiv}
A.~Wyner and J.~Ziv, ``The rate-distortion function for source coding with side
  information at the decoder,'' \emph{IEEE Trans. on Information Theory},
  vol.~22, no.~1, pp. 1 -- 10, jan 1976.

\bibitem{OSullivan2002}
J.~O'Sullivan and N.A.Schmid, ``Large deviations performance analysis for
  biometrics recognition,'' in \emph{Allerton Conference on Communication,
  Control, and Computing}, October 2002.

\bibitem{Willems03}
F.~Willems, T.~Kalker, S.~Baggen, and J.~paul Linnartz, ``On the capacity of a
  biometrical identification system,'' in \emph{In: Proc. of the 2003 IEEE Int.
  Symp. on Inf. Theory}, 2003, pp. 8--2.

\bibitem{Westover08}
M.~Westover and J.~O'Sullivan, ``Achievable rates for pattern recognition,''
  \emph{Information Theory, IEEE Transactions on}, vol.~54, no.~1, pp. 299
  --320, jan. 2008.

\bibitem{Tuncel09}
E.~Tuncel, ``Capacity/storage tradeoff in high-dimensional identification
  systems,'' \emph{Information Theory, IEEE Transactions on}, vol.~55, no.~5,
  pp. 2097 --2106, may 2009.

\bibitem{Tuncel12submitted}
E.~Tuncel and D.~G\"und\"uz, ``Identification and lossy reconstruction in noisy
  databases,'' \emph{IEEE Transactions on Information Theory, to appear}, 2013.

\bibitem{hashingBook}
A.~G. Konheim, \emph{Hashing in Computer Science: Fifty Years of Slicing and
  Dicing}.\hskip 1em plus 0.5em minus 0.4em\relax John Wiley \& Sons, 2010.

\bibitem{Bloom1970}
B.~H. Bloom, ``Space/time trade-offs in hash coding with allowable errors,''
  \emph{Commun. ACM}, vol.~13, no.~7, pp. 422--426, Jul. 1970.

\bibitem{Porat09_Opt_Bloom_Filter_matrix}
E.~Porat, ``An optimal bloom filter replacement based on matrix solving,'' in
  \emph{CSR}, ser. Lecture Notes in Computer Science, A.~E. Frid, A.~Morozov,
  A.~Rybalchenko, and K.~W. Wagner, Eds., vol. 5675.\hskip 1em plus 0.5em minus
  0.4em\relax Springer, 2009, pp. 263--273.

\bibitem{AndoniI08}
A.~Andoni and P.~Indyk, ``Near-optimal hashing algorithms for approximate
  nearest neighbor in high dimensions,'' \emph{Commun. ACM}, vol.~51, no.~1,
  pp. 117--122, 2008.

\bibitem{JohnsonLindenstrauss84}
W.~B. Johnson and J.~Lindenstrauss, ``Extensions of lipschitz mappings into a
  hilbert space,'' \emph{Conf. in Modern Analysis and Probability}, vol.~26,
  pp. 189--206, 1984, conf. was held in 1982, book publ. 1984.

\bibitem{sketchingNotes}
P.~Indyk, \emph{Sketching, streaming and sublinear-space algorithms}.\hskip 1em
  plus 0.5em minus 0.4em\relax Lecture Notes, 2007, Mass. Inst. of Tech.,
  available at http://stellar.mit.edu/S/course/6/fa07/6.895/.

\bibitem{VAfile}
R.~Weber, H.-J. Schek, and S.~Blott, ``A quantitative analysis and performance
  study for similarity-search methods in high-dimensional spaces,'' in
  \emph{VLDB}, vol.~98, 1998, pp. 194--205.

\bibitem{VQfile}
S.~Ramaswamy and K.~Rose, ``Adaptive cluster-distance bounding for nearest
  neighbor search in image databases,'' in \emph{IEEE International Conference
  on Image Processing (ICIP)}, vol.~6.\hskip 1em plus 0.5em minus 0.4em\relax
  IEEE, 2007, pp. VI--381.

\bibitem{semantichashing}
R.~Salakhutdinov and G.~Hinton, ``Semantic hashing,'' \emph{RBM}, vol. 500,
  no.~3, p. 500, 2007.

\bibitem{spectralhashing}
Y.~Weiss, A.~Torralba, and R.~Fergus, ``Spectral hashing,'' in \emph{Advances
  in neural information processing systems}, 2008, pp. 1753--1760.

\bibitem{ragiskyHashing}
M.~Raginsky and S.~Lazebnik, ``Locality-sensitive binary codes from
  shift-invariant kernels,'' in \emph{Advances in neural information processing
  systems}, 2009, pp. 1509--1517.

\bibitem{GallagerInfoTheoryBook}
R.~G. Gallager, \emph{Information Theory and Reliable Communication}.\hskip 1em
  plus 0.5em minus 0.4em\relax New York, NY, USA: John Wiley \& Sons, Inc.,
  1968.

\bibitem{gray1975rhoBar}
R.~M. Gray, D.~L. Neuhoff, and P.~C. Shields, ``A generalization of ornstein's
  d distance with applications to information theory,'' \emph{The Annals of
  Probability}, pp. 315--328, 1975.

\bibitem{villani2009optimal}
C.~Villani, \emph{Optimal transport: old and new}.\hskip 1em plus 0.5em minus
  0.4em\relax Springer, 2009, vol. 338.

\bibitem{marton1996bounding}
K.~Marton, ``Bounding $\bar{d}$-distance by informational divergence: a method
  to prove measure concentration,'' \emph{The Annals of Probability}, vol.~24,
  no.~2, pp. 857--866, 1996.

\bibitem{RaginskySasonConcentration}
M.~Raginsky and I.~Sason, ``Concentration of measure inequalities in
  information theory, communications and coding,'' \emph{CoRR}, vol.
  abs/1212.4663, 2012.

\bibitem{IdoiaAllerton2013}
I.~Ochoa, A.~Ingber, and T.~Weissman, ``Efficient similarity queries via lossy
  compression,'' in \emph{51st Annual Allerton Conference on Communication,
  Control and Computing}, Monticelo, IL, Sep. 2013.

\bibitem{CsiszarKorner2nd}
I.~Csisz\'ar and J.~Korner, \emph{Information Theory - Coding Theorems for
  Discrete Memoryless Systems}.\hskip 1em plus 0.5em minus 0.4em\relax
  Cambridge, 2011.

\bibitem{Jana2009Cardinality}
S.~Jana, ``Alphabet sizes of auxiliary random variables in canonical inner
  bounds,'' in \emph{43rd Annual Conference on Information Sciences and Systems
  (CISS)}, 2009, pp. 67--71.

\bibitem{cvx}
M.~Grant and S.~Boyd, ``{CVX}: Matlab software for disciplined convex
  programming, version 2.0 beta,'' \url{http://cvxr.com/cvx}, Sep. 2013.

\bibitem{WangIngberKochmanJSCCstrongConverse}
D.~Wang, A.~Ingber, and Y.~Kochman, ``A strong converse for joint
  source-channel coding,'' in \emph{Proc. IEEE International Symposium on
  Information Theory}, 2012, pp. 2117--2121.

\bibitem{DemboWeissman2003}
A.~Dembo and T.~Weissman, ``The minimax distortion redundancy in noisy source
  coding,'' \emph{IEEE Trans. on Information Theory}, vol.~49, no.~11, pp.
  3020--3030, 2003.

\bibitem{idoiaDCC2014}
I.~Ochoa, A.~Ingber, and T.~Weissman, ``Compression schemes for similarity
  queries,'' submitted to the 2014 Data Compression Conference.

\bibitem{salehi1978cardinality}
M.~Salehi, ``Cardinality bounds on auxiliary variables in multiple-user theory
  via the method of ahlswede and k{\"o}rner,'' \emph{Dept. Statistics, Stanford
  Univ., Stanford, CA, Tech. Rep}, vol.~33, 1978.

\bibitem{Egglston}
H.~G. Eggleston, J.~A. Todd, and F.~Smithies, \emph{Convexity}.\hskip 1em plus
  0.5em minus 0.4em\relax Cambridge University Press, 1966.

\bibitem{rockafellar70}
R.~T. Rockafellar, \emph{Convex Analysis}.\hskip 1em plus 0.5em minus
  0.4em\relax Princeton University Press, 1970.

\bibitem{IngberKochmanSourceDispersion2011}
A.~Ingber and Y.~Kochman, ``The dispersion of lossy source coding,'' in
  \emph{Proc. of the Data Compression Conference}, Snowbird, Utah, March 2011.

\bibitem{KostinaVerduFiniteBlock}
V.~Kostina and S.~Verd{\'u}, ``Fixed-length lossy compression in the finite
  blocklength regime,'' \emph{IEEE Trans. on Information Theory}, vol.~58,
  no.~6, pp. 3309--3338, 2012.

\bibitem{csiszar1998method}
I.~Csisz{\'a}r, ``The method of types [information theory],'' \emph{IEEE Trans.
  on Information Theory}, vol.~44, no.~6, pp. 2505--2523, 1998.

\bibitem{shannon1959coding}
C.~E. Shannon, ``Coding theorems for a discrete source with a fidelity
  criterion,'' \emph{IRE Nat. Conv. Rec}, vol.~4, no. 142-163, 1959.

\end{thebibliography}

\end{document}